\documentclass[journal,romanappendices]{IEEEtran}
\usepackage{enumitem}
\usepackage{mathpple}
\usepackage{times}
\usepackage{amsmath,amsthm}  
\usepackage{amssymb}  
\usepackage[mathscr]{eucal}
\usepackage{cite}     
\usepackage{comment}  
\usepackage{upref}
\usepackage{amsfonts}
\usepackage{verbatim}
\usepackage[dvipsnames,usenames]{color}
\usepackage{enumerate}
\usepackage{graphicx}
\usepackage{latexsym}
\usepackage{bm}
\usepackage{multirow}
\usepackage{dsfont}
\usepackage{amsthm,xpatch}
\usepackage{mathtools}
\usepackage{bbm}
\usepackage{latexsym}
\usepackage{accents}
\usepackage{eufrak}
\usepackage{psfrag}
\usepackage{color}
\usepackage{times,color}
\usepackage[usenames,dvipsnames,svgnames,table]{xcolor}
\usepackage{pgf}
\usepackage{tikz}
\usetikzlibrary{arrows, automata, positioning, calc}
\usepackage{cancel} 
\usepackage{caption}
\usepackage{subcaption}
\usepackage{rotating}
\usepackage[T1]{fontenc}
\usepackage{anyfontsize}
\usepackage[hyphens]{url}
\usepackage{hyperref}

\newcommand\bovermat[2]{%
  \makebox[0pt][l]{$\smash{\overbrace{\phantom{%
    \begin{matrix}#2\end{matrix}}}^{\text{#1}}}$}#2}

\newcommand\bundermat[2]{%
  \makebox[0pt][l]{$\smash{\underbrace{\phantom{%
    \begin{matrix}#2\end{matrix}}}_{\text{#1}}}$}#2}

\newcommand{\myLambda}{\begin{sideways}%
     \begin{sideways}$\mathrm{V}$\end{sideways}\end{sideways}}

\newtheorem{lemma}{Lemma}
\newtheorem{theorem}{Theorem}
\newtheorem{corollary}{Corollary}
\newtheorem{remark}{Remark}

\usepackage{algorithm,algpseudocode}

\makeatletter
\newcommand{\removelatexerror}{\let\@latex@error\@gobble}
\makeatletter
\newcommand{\proofpart}[2]{%
	\par
	\addvspace{\medskipamount}%
	\noindent\emph{Part #1: #2}\par\nobreak
	\addvspace{\smallskipamount}%
	\@afterheading
}
\makeatother

\makeatletter
\newcommand*{\transpose}{%
  {\mathpalette\@transpose{}}%
}
\newcommand*{\@transpose}[2]{%
  \raisebox{\depth}{$\m@th#1\intercal$}%
}
\makeatother

\renewcommand{\mathsf}[1]{#1}

\theoremstyle{definition}
\newtheorem{example}{Example}

\IEEEoverridecommandlockouts

\begin{document}

\newcommand{\SB}[3]{
\sum_{#2 \in #1}\biggl|\overline{X}_{#2}\biggr| #3
\biggl|\bigcap_{#2 \notin #1}\overline{X}_{#2}\biggr|
}

\newcommand{\Mod}[1]{\ (\textup{mod}\ #1)}

\newcommand{\overbar}[1]{\mkern 0mu\overline{\mkern-0mu#1\mkern-8.5mu}\mkern 6mu}

\makeatletter
\newcommand*\nss[3]{%
  \begingroup
  \setbox0\hbox{$\m@th\scriptstyle\cramped{#2}$}%
  \setbox2\hbox{$\m@th\scriptstyle#3$}%
  \dimen@=\fontdimen8\textfont3
  \multiply\dimen@ by 4             
  \advance \dimen@ by \ht0
  \advance \dimen@ by -\fontdimen17\textfont2
  \@tempdima=\fontdimen5\textfont2  
  \multiply\@tempdima by 4
  \divide  \@tempdima by 5          
  \ifdim\dimen@<\@tempdima
    \ht0=0pt                        
    \@tempdima=\fontdimen5\textfont2
    \divide\@tempdima by 4          
    \advance \dimen@ by -\@tempdima 
    \ifdim\dimen@>0pt
      \@tempdima=\dp2
      \advance\@tempdima by \dimen@
      \dp2=\@tempdima
    \fi
  \fi
  #1_{\box0}^{\box2}%
  \endgroup
  }
\makeatother

\makeatletter
\renewenvironment{proof}[1][\proofname]{\par
  \pushQED{\qed}%
  \normalfont \topsep6\p@\@plus6\p@\relax
  \trivlist
  \item[\hskip\labelsep
        \itshape
    #1\@addpunct{:}]\ignorespaces
}{%
  \popQED\endtrivlist\@endpefalse
}
\makeatother

\makeatletter
\newsavebox\myboxA
\newsavebox\myboxB
\newlength\mylenA

\newcommand*\xoverline[2][0.75]{%
    \sbox{\myboxA}{$\m@th#2$}%
    \setbox\myboxB\null
    \ht\myboxB=\ht\myboxA%
    \dp\myboxB=\dp\myboxA%
    \wd\myboxB=#1\wd\myboxA
    \sbox\myboxB{$\m@th\overline{\copy\myboxB}$}
    \setlength\mylenA{\the\wd\myboxA}
    \addtolength\mylenA{-\the\wd\myboxB}%
    \ifdim\wd\myboxB<\wd\myboxA%
       \rlap{\hskip 0.5\mylenA\usebox\myboxB}{\usebox\myboxA}%
    \else
        \hskip -0.5\mylenA\rlap{\usebox\myboxA}{\hskip 0.5\mylenA\usebox\myboxB}%
    \fi}
\makeatother

\xpatchcmd{\proof}{\hskip\labelsep}{\hskip3.75\labelsep}{}{}

\pagestyle{plain}

\title{\fontsize{21}{28}\selectfont Single-Server Private Linear Transformation:\\ The Individual Privacy Case}
\author{Anoosheh Heidarzadeh, \emph{Member, IEEE}, Nahid Esmati, \emph{Student Member, IEEE}, \\ and Alex Sprintson, \emph{Senior Member, IEEE}\thanks{This work is to be presented in part at the 2021 IEEE International Symposium on Information Theory, Melbourne, Australia, July 2021.}\thanks{The authors are with the Department of Electrical and Computer Engineering, Texas A\&M University, College Station, TX 77843 USA (E-mail: \{anoosheh, nahid, spalex\}@tamu.edu).}}


%


\maketitle 

\thispagestyle{plain}
\begin{abstract}


This paper considers the single-server Private Linear Transformation (PLT) problem with  individual privacy guarantees. In this problem, there is a user that wishes to obtain $L$ independent linear combinations of a $D$-subset of messages belonging to a dataset of $K$ messages stored on a single server. The goal is to minimize the download cost while keeping the identity of each message required for the computation individually private. The individual privacy requirement ensures that the identity of each individual message required for the computation is kept private. 
This is in contrast to the stricter notion of joint privacy that protects the entire set of identities of all messages used for the computation, including the correlations between these identities. The notion of individual privacy captures a broad set of practical applications. For example, such notion is relevant when the dataset contains information about individuals, each of them requires privacy guarantees for their data access patterns. 
    
We focus on the setting in which the required linear transformation is associated with a maximum distance separable (MDS) matrix. In particular, we require that the matrix of coefficients pertaining to the required linear combinations is the generator matrix of an MDS code. We establish lower and upper bounds on the capacity of PLT with individual privacy, where the capacity is defined as the supremum of all achievable download rates. We show that our bounds are tight under certain conditions. 
\end{abstract}


\section{introduction}
In this work, 
we study the problem of single-server  \emph{Private Linear Transformation (PLT) with individual privacy}, referred to as \emph{IPLT} for short. 
In this problem, there is a single server that stores a set of $K$ messages, and a user that wants to compute $L$ independent linear combinations of a subset of $D$ messages. 
The objective of the user is to recover the required linear combinations by downloading minimum possible amount of information from the server, while 
protecting the identity of each message required for the computation individually.
More specifically, the individual privacy requirement implies that, from the server's perspective, every message must be equally likely \emph{a posteriori} to belong to the support set of the $L$ required linear combinations, assuming that all $D$-subsets of messages are \emph{a priori} equiprobable to be the support set of the $L$ required linear combinations.

This setup appears in several practical scenarios including Machine Learning (ML) applications such as linear transformation for dimensionality reduction~\cite{CG2015}, and parallel training of different linear regression or classification models~\cite{A2019,JY2020}. 
For example, consider a scenario in which the server stores a dataset with $N$ data samples each with $K$ attributes, and the $N$ data samples for each attribute represent one message. 
The user would like to run an ML algorithm on a subset of $D$ selected attributes, but they wish to hide the identity of each of the selected attributes individually.
For instance, each attribute may correspond to an individual, and the user is required to hide from the server whether the information belonging to an individual was used. 
When $D$ is large, it is beneficial is to reduce the $D$-dimensional feature space into a smaller space of dimension $L$. 
This dimensionality reduction can be performed by a linear transformation. 
In this case, instead of retrieving the $D$ messages corresponding to the selected attributes, only $L$ linear combinations need to be retrieved.
Retrieving these linear combinations while protecting the privacy of each of the selected attributes, matches the setup of the IPLT problem.

The notion of individual privacy was originally introduced in~\cite{HKRS2019} for single-server Private Information Retrieval (PIR) with individual privacy guarantees (or IPIR), and was recently considered for single-server Private Linear Computation (PLC) with individual privacy guarantees (or IPLC) in~\cite{HS2020}.
The IPLT problem generalizes the IPIR and IPLC problems. 
In particular, the IPLT problem reduces to the IPIR problem or the IPLC problem when $L=D$ or $L=1$, respectively.
The IPLT problem is also related to the problem of single-server PLT with joint privacy guarantees (or JPLT for short), which we have studied in a parallel work~\cite{HES2021JointJournal}. 
The notion of joint privacy was previously considered for the problems of multi-message PIR~\cite{BU2018,HKGRS2018,LG2018,MMM2019} and single-server PLC~\cite{HS2019PC,HS2020}.
The joint privacy condition implies that, 
from the server's perspective, 
every $D$-subset of messages must be equally likely \emph{a posteriori} to be the support set of the $L$ required linear combinations. 
It is easy to see that individual privacy is weaker than joint privacy. 
That said, individual privacy has an interesting operational meaning per se, and is motivated by the need to protect the access pattern for individual (rather than the entire set of) messages required for the computation. 

The joint and individual privacy guarantees are applicable to the scenarios in which the data access patterns need to be protected. 
Note that these types of access privacy are different from the privacy requirements for the multi-server PLC problem in~\cite{SJ2018,MM2018,OK2018,OLRK2018} and the multi-server Private Monomial Computation problem~\cite{YLR2020}. 
In particular, the privacy requirement in~\cite{SJ2018,MM2018,OK2018,OLRK2018} is to hide the values of the combination coefficients in the required linear combination; and 
the privacy requirement in~\cite{YLR2020} is to hide the values of the exponents in the required monomial function. 
The IPIR and IPLC problems were previously studied in the settings in which the user has a prior side information about a subset of messages. 
As was shown in~\cite{HKRS2019,HS2020}, when compared to single-server PIR with joint privacy guarantees (or JPIR) and PLC with joint privacy guarantees (or JPLC)~\cite{HKGRS2018,LG2018,HS2019PC}, IPIR and IPLC can be performed with a much lower download cost. 
Motivated by these results, this work seeks to answer the following questions:
(i) when there is no prior side information, is it possible to perform IPLT with a lower download cost than JPLT? 
(ii) what are the fundamental limits on the download cost for IPLT? 
In this work, we make a significant progress towards answering these questions.

\subsection{Main Contributions}


In this work, we focus on the setting in which the coefficient matrix corresponding to the required linear combinations is the generator matrix of a maximum distance separable (MDS) code. 
The MDS matrices are motivated by the scenarios in which the combination coefficients are judiciously chosen to form an MDS matrix, 
or they are randomly generated over the field of reals (or a sufficiently large finite field), and form an MDS matrix with probability $1$ (or with high probability). 

We establish bounds on the capacity of IPLT, where the capacity is defined as the supremum of all achievable download rates. 
In particular, we prove an upper bound on the capacity using a novel converse proof technique which relies on several linear-algebraic and information-theoretic arguments. 
Using this technique, we formulate the problem of upper bounding the capacity as an integer linear programming (ILP) problem. 
Solving this ILP, we obtain an upper bound on the capacity. 
We also prove a lower bound on the capacity by designing an achievability scheme, termed \emph{Generalized Partition-and-Code with Partial Interference Alignment (GPC-PIA) protocol}. 
This protocol generalizes the protocols we recently proposed in~\cite{HKRS2019} and~\cite{HS2020} for the IPIR problem and the IPLC problem, respectively.
In addition, we show that our bounds are tight under certain conditions, particularly if $R=K \pmod D\leq L$, or $R$ divides $D$, settling the capacity of IPLT for such cases. 
Our results show that 
(i) for a wide range of values of $K,D,L$, the capacity of IPLT is higher than that of JPLT, 
i.e., IPLT can be performed more efficiently than JPLT in terms of the download cost; and 
(ii) for some other range of values of $K,D,L$, the capacity of IPLT and JPLT are the same, 
i.e., IPLT is as costly as JPLT in terms of the download cost.   



\subsection{Notation}
Throughout, we denote random variables and their realizations by bold-face symbols and regular symbols, respectively. 
We also denote sets, vectors, and matrices by roman font, and collections of sets, vectors, or matrices by blackboard bold roman font. 
For any random variables $\mathbf{X},\mathbf{Y}$, we denote the entropy of $\mathbf{X}$ and the conditional entropy of $\mathbf{X}$ given $\mathbf{Y}$ by $H(\mathbf{X})$ and $H(\mathbf{X}|\mathbf{Y})$ respectively. 
For any integer $n\geq 1$, we denote $\{1,\dots,n\}$ by $[n]$, and for any integers $1<n<m$, we denote $\{n,\dots,m\}$ by $[n:m]$.  
We denote the binomial coefficient $\binom{n}{k}$ by $C_{n,k}$.
For any positive integers $a,b$, we write $a\mid b$ (or $a\nmid b$) if $a$ divides $b$ (or $a$ does not divide $b$).







\section{Problem Setup}\label{sec:SN}






\subsection{Models and Assumptions}
Let $q$ be an arbitrary prime power, and let $N\geq 1$ be an arbitrary integer. 
Let $\mathbbmss{F}_q$ be a finite field of order $q$, and let $\mathbbmss{F}_{q}^{N}$ be the $N$-dimensional vector space over $\mathbbmss{F}_q$.
Let $B\triangleq N\log_2 q$. 
Let $K,D,L\geq 1$ be integers such that ${L\leq D\leq K}$. 
We denote by $\mathbbm{W}$ the set of all $D$-subsets of $[K]$, and denote by $\mathbbm{V}$ the set of all $L\times D$ MDS matrices with entries in $\mathbbmss{F}_q$.
\footnote{For any $1\leq k\leq n$, a $k\times n$ matrix $\mathrm{M}$ is said to be maximum distance separable (MDS) if $\mathrm{M}$ generates an $[n,k]$ MDS code. 
Equivalently, a $k\times n$ matrix $\mathrm{M}$ is said to be MDS if every $k\times k$ submatrix of $\mathrm{M}$ is invertible.} 

Suppose that there is a server that stores $K$ messages ${X_1,\dots,X_K}$, where $X_i\in \mathbbmss{F}_q^{N}$ for $i\in [K]$ is a row-vector of length $N$. 
Let ${\mathrm{X}\triangleq [X_1^{\transpose},\dots,X_K^{\transpose}]^{\transpose}}$ be the $K\times N$ matrix of messages.
For every ${\mathrm{S}\subset [K]}$, we denote by $\mathrm{X}_{\mathrm{S}}$ the matrix $\mathrm{X}$ restricted to its rows indexed by $\mathrm{S}$. 
Suppose that there is a user who wants to compute 
the $L\times N$ matrix $\mathrm{Z}^{[\mathrm{W},\mathrm{V}]}\triangleq \mathrm{V}\mathrm{X}_{\mathrm{W}}$, where ${\mathrm{W}\in \mathbbm{W}}$ and ${\mathrm{V}\in \mathbbm{V}}$. 
The $L$ rows of the matrix $\mathrm{Z}^{[\mathrm{W},\mathrm{V}]}$ are given by $\mathrm{v}_1 \mathrm{X}_{\mathrm{W}},\dots,\mathrm{v}_L \mathrm{X}_{\mathrm{W}}$, where $\mathrm{v}_l$ for $l\in [L]$ is the $l$th row of the $L\times D$ matrix $\mathrm{V}$, i.e., $\mathrm{V}=[\mathrm{v}^{\transpose}_1,\dots,\mathrm{v}^{\transpose}_L]^{\transpose}$. 
Note that $\mathrm{V} \mathrm{X}_{\mathrm{W}}$ corresponds to $L$ MDS coded linear combinations of the $D$ messages indexed by $\mathrm{W}$ where the combination coefficients are specified by the MDS matrix $\mathrm{V}$.
We refer to $\mathrm{Z}^{[\mathrm{W},\mathrm{V}]}$ as the \emph{demand}, $\mathrm{W}$ as the \emph{support of the demand}, $\mathrm{V}$ as the \emph{coefficient matrix of the demand}, $D$ as the \emph{support size of the demand}, and $L$ as the \emph{dimension of the demand}.

In this work, we assume that (i) $\mathbf{X}_1,\dots,\mathbf{X}_K$ are independently and uniformly distributed over $\mathbbmss{F}_{q}^{N}$. 
This implies that $H(\mathbf{X})=KB$, ${H(\mathbf{X}_{\mathrm{S}})= |\mathrm{S}| B}$ for every ${\mathrm{S}\subset [K]}$, and $H(\mathbf{Z}^{[\mathrm{W},\mathrm{V}]})=LB$; (ii) the random variables $\mathbf{W}, \mathbf{V}, \mathbf{X}$ are independent; (iii) $\mathbf{W}$ and $\mathbf{V}$ are distributed uniformly over ${\mathbbm{W}}$ and $\mathbbm{V}$, respectively; (iv) the parameters $D$ and $L$, and the distribution of $(\mathbf{W},\mathbf{V})$ are initially known by the server; and (v) the realization $(\mathrm{W},\mathrm{V})$ is not initially known by the server.



\subsection{Privacy and Recoverability Conditions}
Given $(\mathrm{W}$,$\mathrm{V})$, the user generates a query $\mathrm{Q} = \mathrm{Q}^{[\mathrm{W},\mathrm{V}]}$, which is a (deterministic or stochastic) function of $(\mathrm{W},\mathrm{V})$, and sends $\mathrm{Q}$ to the server. 
The query $\mathrm{Q}$ 
must satisfy the following privacy condition: given the query $\mathrm{Q}$, 
every individual message index must be equally likely to belong to the demand's support. 
That is, for every $i\in [K]$, it must hold that
\begin{equation*}
\Pr (i\in \mathbf{W}|\mathbf{Q}=\mathrm{Q})=\Pr(i \in \mathbf{W})=D/K, 
\end{equation*} where $\mathbf{Q}$ denotes $\mathbf{Q}^{[\mathbf{W},\mathbf{V}]}$. 
This condition---which was recently introduced in~\cite{HKRS2019} and~\cite{HS2020} for single-server PIR and PLC, is referred to as the \emph{individual privacy condition}. 

Upon receiving the query $\mathrm{Q}$, the server generates an answer $\mathrm{A}=\mathrm{A}^{[\mathrm{W},\mathrm{V}]}$, 
and sends it back to the user. 
The answer $\mathrm{A}$ is a deterministic function of $\mathrm{Q}$ and $\mathrm{X}$.
That is, $H(\mathbf{A}|\mathbf{Q},\mathbf{X})=0$, where $\mathbf{A}$ denotes $\mathbf{A}^{[\mathbf{W},\mathbf{V}]}$. 
The collection of the answer $\mathrm{A}$, the query $\mathrm{Q}$, and the realization $(\mathrm{W},\mathrm{V})$, must enable the user to recover the demand $\mathrm{Z}^{[\mathrm{W},\mathrm{V}]}$. 
That is, 
\[{H(\mathbf{Z}| \mathbf{A},\mathbf{Q}, \mathbf{W},\mathbf{V})=0},\] where $\mathbf{Z}$ denotes $\mathbf{Z}^{[\mathbf{W},\mathbf{V}]}$. This condition is referred to as the \emph{recoverability condition}. 
 
\subsection{Problem Statement} 
We would like to design a protocol for generating a query $\mathrm{Q}^{[\mathrm{W},\mathrm{V}]}$ and the corresponding answer $\mathrm{A}^{[\mathrm{W},\mathrm{V}]}$ for any given $(\mathrm{W},\mathrm{V})$, 
such that the individual privacy and recoverability conditions are satisfied.
We refer to this problem as single-server \emph{Private Linear Transformation (PLT) with Individual Privacy}, or \emph{IPLT} for short. 

We say that a protocol is \emph{deterministic} (or \emph{randomized}) if the user's query is a deterministic (or stochastic) function of $(\mathrm{W}$,$\mathrm{V})$. 
Also, we say that a protocol is \emph{linear} if the server's answer contains only linear combinations of the messages; otherwise, the protocol is said to be \emph{non-linear}.

Following the convention in the PIR and PLC literature, we define the \emph{rate} of an IPLT protocol as the ratio of the entropy of the demand (i.e., $H(\mathbf{Z})=LB$) to the entropy of the answer (i.e., $H(\mathbf{A})$). 
We also define the \emph{capacity} of IPLT as the supremum of rates over all IPLT protocols and over all field sizes $q$. 
In this work, our goal is to establish (tight) bounds (in terms of $K, D, L$) on the capacity of IPLT. 


\section{Main Results}
In this section, we summarize our main results on the capacity of the IPLT setting. 
Theorems~\ref{thm:IPLT-UB} and~\ref{thm:IPLT-LB} present an upper bound and a lower bound on the capacity, respectively, and 
Corollary~\ref{cor:IPLT-Cap} characterizes the capacity under certain conditions, depending on $K,D,L$. 
The proofs of Theorems~\ref{thm:IPLT-UB} and~\ref{thm:IPLT-LB} are given in Sections~\ref{sec:IPLT-Conv} and~\ref{sec:IPLT-Ach}, respectively.   

For simplifying the notation, we define $R\triangleq K \pmod D$ and $S\triangleq \gcd(D+R,R)$, and use the notations $R,S$ along with the basic notations $K,D,L$ everywhere.  

\begin{theorem}\label{thm:IPLT-UB}
For the IPLT setting with $K$ messages, demand's support size $D$, and demand's dimension $L$, the capacity is upper bounded by 
\begin{equation}\label{eq:UB}
\left(\left\lfloor \frac{K}{D}\right\rfloor+\min\left\{1,\frac{R}{L}\right\}\right)^{-1}.    
\end{equation} 
\end{theorem}


To prove this result, we a mix of information-theoretic and linear-algebraic arguments which rely on the individual privacy and recoverability conditions, and form an integer linear programming (ILP) problem. 
Solving this ILP, we obtain the upper bound~\eqref{eq:UB} on the capacity. 

\begin{theorem}\label{thm:IPLT-LB}
For the IPLT setting with $K$ messages, demand's support size $D$, and demand's dimension $L$, the capacity is lower bounded by
\begin{equation}\label{eq:LB}
\left(\left\lfloor \frac{K}{D}\right\rfloor+\min\left\{\frac{R}{S},\frac{R}{L}\right\}\right)^{-1}.    
\end{equation}
\end{theorem}

We prove the lower bound~\eqref{eq:LB} on the capacity by constructing an IPLT protocol, termed \emph{Generalized Partition-and-Code with Partial Interference Alignment (GPC-PIA)}. 
This protocol generalizes the protocols we previously proposed in~\cite{HKRS2019,HS2020} for the IPIR and IPLC problems. 
The main ingredients of the GPC-PIA protocol are as follows: 
(i)~constructing a properly designed family of subsets of messages, where some subsets are possibly overlapping, and 
(ii)~designing a number of linear combinations for each subset, where the linear combinations pertaining to the overlapping subsets are partially aligned.

\begin{corollary}\label{cor:IPLT-Cap}
For the IPLT setting with $K$ messages, demand's support size $D$, and demand's dimension $L$, if $R\leq L$ or $R\mid D$, the capacity is given by \[\left(\left\lfloor \frac{K}{D}\right\rfloor+\min\left\{1,\frac{R}{L}\right\}\right)^{-1}.\] 
In particular, if $D\mid K$, the capacity is given by ${D}/{K}$.
\end{corollary}

\begin{proof}
The result follows immediately from comparing the upper and lower bounds in Theorems~\ref{thm:IPLT-UB} and~\ref{thm:IPLT-LB}.
\end{proof}

\begin{figure*}[t!]
\centering
\begin{subfigure}[t]{.5\textwidth}
  \centering
  \includegraphics[width=.75\linewidth]{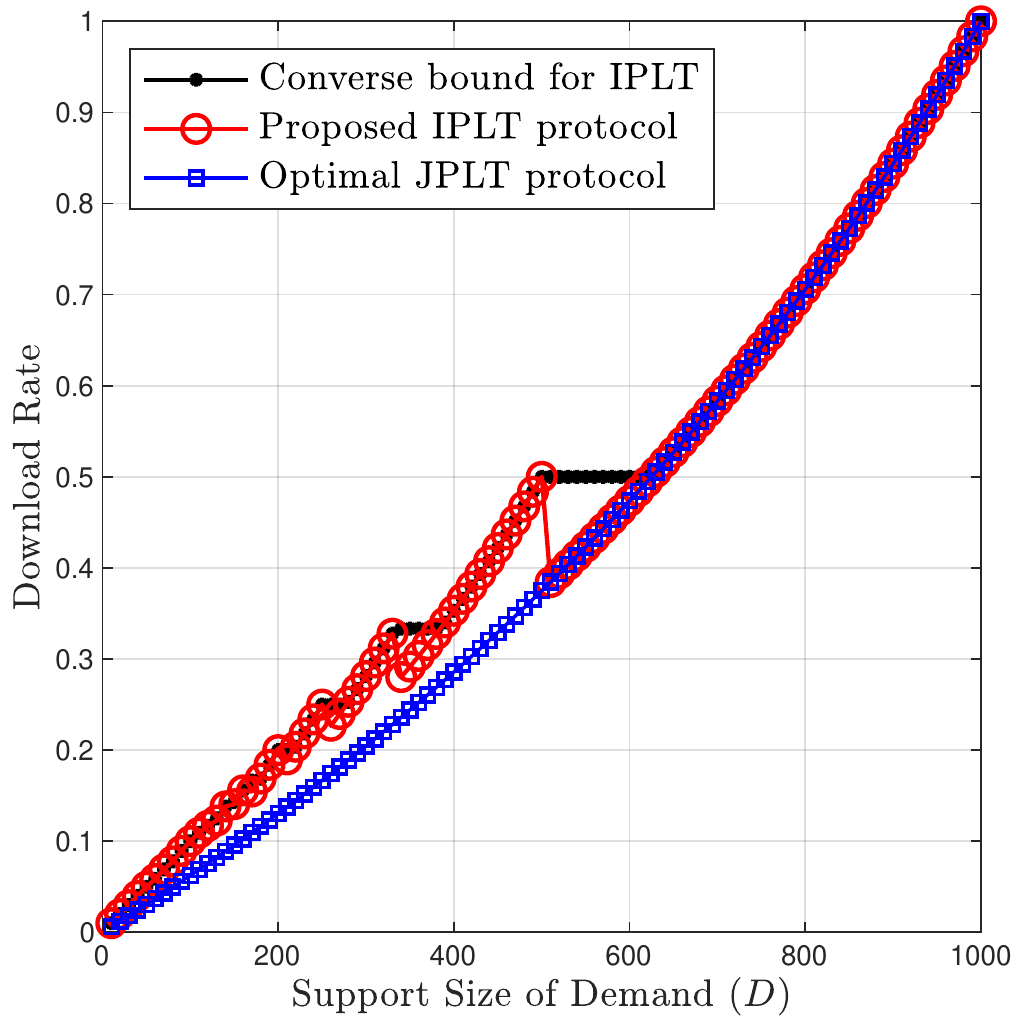}
  \caption{$K=1000$, $L/D=0.6$}
  \label{fig:IPLT1}
\end{subfigure}%
\begin{subfigure}[t]{.5\textwidth}
  \centering
  \includegraphics[width=.75\linewidth]{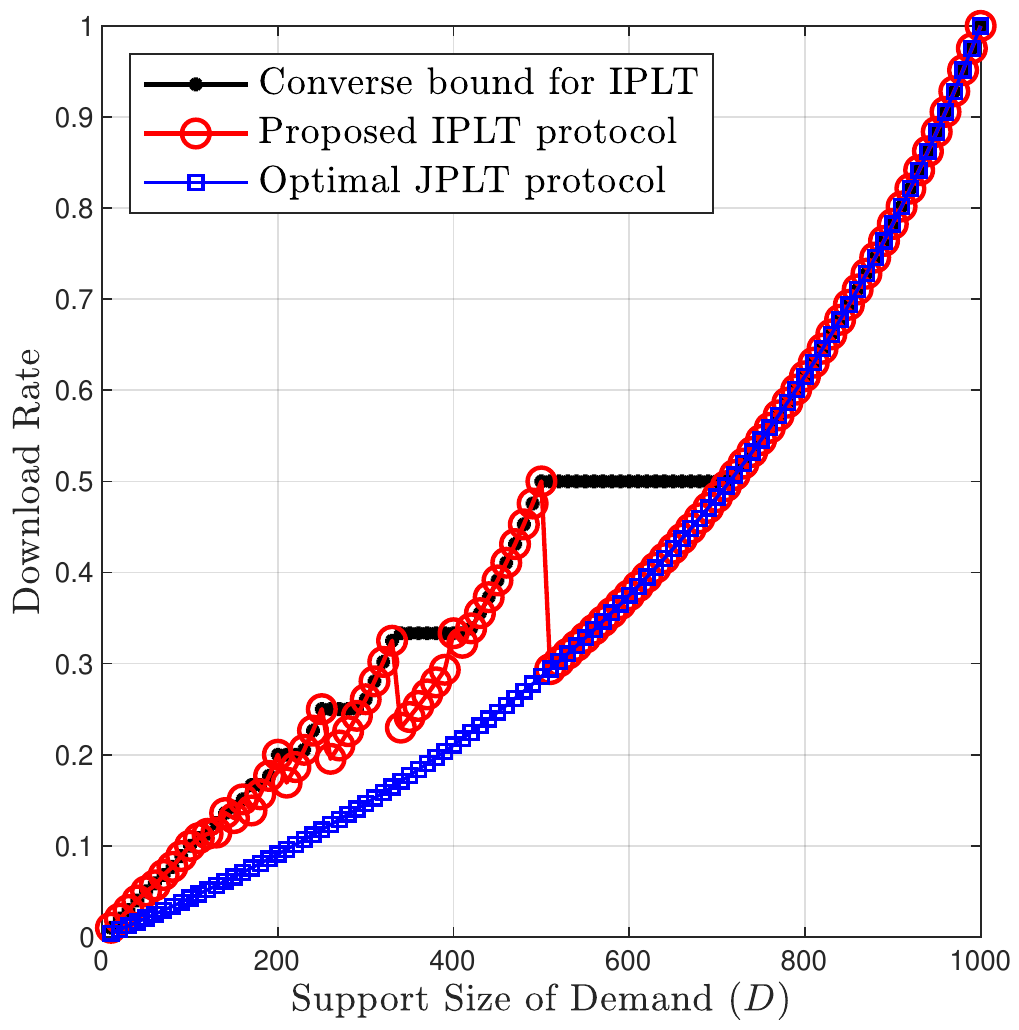}
  \caption{$K=1000$, $L/D=0.4$}
  \label{fig:IPLT2}
\end{subfigure}
\caption{The download rate of the proposed IPLT protocol and the optimal JPLT protocol of~\cite{HES2021JointJournal}.}
\label{fig:IPLT}
\end{figure*}

\begin{remark}\label{rem:IPLT1}
\emph{As shown in~\cite{HS2020}, the capacity of IPLC with side information is given by $\lceil {K}/{(D+M)}\rceil^{-1}$, where the user initially knows ${M\geq 1}$ uncoded messages or one linear combination of ${M\geq 1}$ messages as side information, and the identities of these $M$ messages are not initially known by the server. 
The capacity of this setting was, however, left open for $M=0$. 
Theorems~\ref{thm:IPLT-UB} and~\ref{thm:IPLT-LB} respectively provide an upper bound ${(\lfloor {K}/{D}\rfloor+\min\{1,R\})^{-1}}$ and a lower bound ${(\lfloor {K}/{D}\rfloor+\min\{{R}/{S},R\})^{-1}}$ on the capacity of this setting as a special case of IPLT for $L=1$. 
Interestingly, these bounds match if $R=0$ or $R\mid D$, settling the capacity of IPLC for $M=0$, when $R=0$ or $R\mid D$. 
For $L=D$, IPLT reduces to IPIR without side information. 
It is known that the optimal download rate in this case 
is 
${D}/{K}$~\cite{HKRS2019}. 
This is consistent with our results. 
Note that for $L=D$, it holds that $R=K \pmod D \leq D=L$, and by the result of Corollary~\ref{cor:IPLT-Cap}, the capacity for this case is given by $(\lfloor {K}/{D} \rfloor+\min\{1,{R}/{D}\})^{-1}=(\lfloor {K}/{D} \rfloor+{R}/{D})^{-1} = ({K}/{D})^{-1}={D}/{K}$. 
}
\end{remark}

\begin{remark}\label{rem:IPLT2}
\emph{Naturally, any JPLT protocol can also serve as an IPLT protocol. 
This comes from the fact that joint privacy is a stricter notion that implies individual privacy. 
As we showed in~\cite{HES2021JointJournal}, an optimal JPLT protocol achieves the rate $L/(K-D+L)$. 
In order to compare the performance of the optimal JPLT protocol of~\cite{HES2021JointJournal} and the proposed IPLT protocol, we depict the download rate of these protocols in Fig.~\ref{fig:IPLT}, for different values of $D\in \{10,20,\dots,1000\}$, where $K=1000$, and $L/D=0.6$ (left plot) or $L/D=0.4$ (right plot). 
One can observe that, when the ratio $L/D$ is fixed, for sufficiently small values of $D$, the download rate of our IPLT protocol is higher than that of the JPLT protocol of~\cite{HES2021JointJournal}; 
whereas, for values of $D$ larger than a threshold, both protocols achieve the same rate. 
This implies that for sufficiently large $D$, achieving individual privacy is as costly as achieving joint privacy. 
In addition, for some values of $D$, the rate achieved by our IPLT protocol matches the converse bound. 
This, in turn, confirms the optimality of our IPLT protocol for such values of $D$. 
By comparing the left and right plots in Fig.~\ref{fig:IPLT}, it can also be seen that for a sufficiently small value of $D$, 
the smaller is the ratio $L/D$, the better is the performance of our IPLT protocol as compared to the JPLT protocol of~\cite{HES2021JointJournal}.
For instance, for $D=250$, the rate of our IPLT protocol is about $33\%$ and $53\%$ more than that of the JPLT protocol of~\cite{HES2021JointJournal} for $L/D=0.6$ and $L/D=0.4$, respectively.
}
\end{remark}

\section{Linear IPLT Protocols and Linear Codes}\label{sec:LINEAR}

While the individual privacy and recoverability conditions must hold for any linear or non-linear IPLT protocol, they establish an interesting connection between linear IPLT protocols and linear codes. 
Below, we discuss this connection for both deterministic and randomized protocols. 
  


Consider a deterministic linear IPLT protocol. 
Consider an arbitrary ordering of all elements in $\mathbbmss{W}$ and $\mathbbmss{V}$, denoted by $\{\mathrm{W}_k\}_{k\in [w]}$ and $\{\mathrm{V}_l\}_{l\in [v]}$, respectively, where $w\triangleq |\mathbbmss{W}|$ and $v\triangleq |\mathbbmss{V}|$.
For any $k\in [w]$ and $l\in [v]$, we denote by $\mathscr{C}_{k,l}$ the corresponding linear code for the instance $(\mathrm{W}_k,\mathrm{V}_l)$. 
That is, $\mathscr{C}_{k,l}$ is the code corresponding to the coefficient matrix of the linear combinations that constitute the answer $\mathrm{A}^{[\mathrm{W}_k,\mathrm{V}_l]}$ to the query $\mathrm{Q}^{[\mathrm{W}_k,\mathrm{V}_l]}$. 
Note that $\mathscr{C}_{k,l}$'s are not necessarily distinct, and $\{\mathscr{C}_{k,l}\}_{k,l}$ is a multiset in general. 
Let $m$ be the number of distinct elements in the multiset $\{\mathscr{C}_{k,l}\}_{k,l}$, and let $\mathscr{C}_1,\dots,\mathscr{C}_m$ and $r_1,\dots,r_m$ be the 
distinct elements and their multiplicities in the multiset $\{\mathscr{C}_{k,l}\}_{k,l}$, respetively.  


For any $k,l$, a linear code $\mathscr{C}$ of length $K$ is said to be \emph{$(k,l)$-feasible} 
if $\mathscr{C}$ contains a collection $\mathrm{C}$ of $L$ codewords whose support is a subset of $\mathrm{W}_{k}$, and the 
code generated by $\mathrm{C}$, when punctured at the coordinates indexed by $\mathrm{W}_{k}$, is identical to the code generated by $\mathrm{V}_{l}$.
\footnote{To puncture a linear code at a coordinate, the column corresponding to that coordinate is deleted from the generator matrix of the code.} 

Note that the $(k,l)$-feasibility is simply a necessary and sufficient condition for recoverability, for the instance $(\mathrm{W}_k,\mathrm{V}_l)$. 
That is, for recoverability, it is necessary and sufficient that for any $k,l$, $\mathscr{C}_{k,l}$ is $(k,l)$-feasible.

Having defined the notion of $(k,l)$-feasibility, a necessary condition for individual privacy is that 
for any ${i\in [K]}$ and ${j\in [m]}$, there exists a pair ${(k,l)\in [w]\times [v]}$ such that 
\begin{itemize}
    \item[1)] $\mathrm{W}_k$ contains the coordinate $i$;
    \item[2)] $\mathscr{C}_{k,l}$ is $(k,l)$-feasible;
    \item[3)] $\mathscr{C}_{k,l}$ and $\mathscr{C}_j$ are identical.
\end{itemize}
To verify the necessity of this condition for individual privacy, 
suppose that for given $i,j$ there is no such pair $(k,l)$. 
Then, if the answer corresponds to the code $\mathscr{C}_j$, the message index $i$ has zero probability to belong to the demand's support.
This obviously violates the individual privacy condition.
However, this necessary condition is not sufficient for individual privacy. 
For any $i,j$, let $n_{i,j}$ be the number of pairs $(k,l)$ such that the conditions 1-3 are satisfied.
Note that $n_{i,j}/r_j$ is equal to the conditional probability that the message index $i$ belongs to the demand's support, given that $\mathscr{C}_j$ is the code corresponding to the answer. 
The above necessary condition for individual privacy simply states that ${n_{i,j}>0}$ for all $i,j$. 
However, there may exist two coordinates $i_1,i_2$ such that $n_{i_1,j}\neq n_{i_2,j}$ for some $j$. 
This asymmetry, in turn, may cause a violation of the individual privacy condition.
A necessary and sufficient condition for individual privacy is that for any $j\in [m]$, ${n_{i,j}=n_j}$ for all $i\in [K]$, for some integer ${n_j>0}$. 

For any $k,l$, let $d_{k,l}$ be the dimension of the code $\mathscr{C}_{k,l}$, and $d_{\text{ave}}$ be the average of $d_{k,l}$'s over all $k,l$. 
The rate of a deterministic linear IPLT protocol is equal to $1/d_{\text{ave}}$. 
Maximizing the rate is then equivalent to minimizing $d_{\text{ave}}$, subject to the above necessary and sufficient conditions for individual privacy and recoverability.




Any randomized linear IPLT protocol can be represented, for any instance $(\mathrm{W}_k,\mathrm{V}_l)$, by a (finite) ensemble of distinct linear codes of length $K$, say, $\mathscr{C}^{1}_{k,l},\dots,\mathscr{C}^{n}_{k,l}$ for some integer $n$ ($=n(k,l)$), and their respective (nonzero) probabilities $p^{1}_{k,l},\dots,p^{n}_{k,l}$, where $\mathscr{C}^{h}_{k,l}$ for $h\in [n]$ is the corresponding code for the instance $(\mathrm{W}_k,\mathrm{V}_l)$ with probability $p^{h}_{k,l}$.
Note that $\sum_{h=1}^{n}p^{h}_{k,l} = 1$. 
Let $m$ be the number of distinct elements in the multiset $\{\mathscr{C}^{h}_{k,l}\}_{k,l,h}$, and let $\mathscr{C}_1,\dots,\mathscr{C}_m$ be the distinct elements in the multiset $\{\mathscr{C}^{h}_{k,l}\}_{k,l,h}$. 



For any $i\in [K]$ and $j\in [m]$, let $q_{i,j}$ be the sum of probabilities $p^{h}_{k,l}$ over all $k,l,h$ such that 
$\mathrm{W}_k$ contains the coordinate $i$, $\mathscr{C}^{h}_{k,l}$ is $(k,l)$-feasible, and $\mathscr{C}^{h}_{k,l}$ and $\mathscr{C}_{j}$ are identical. 
For any $j\in [m]$, let $r_j$ be the sum of probabilities $p^{h}_{k,l}$ over all $k,l,h$ such that $\mathscr{C}^{h}_{k,l}$ and $\mathscr{C}_{j}$ are identical. 
Note that $q_{i,j}/r_j$ is the conditional probability that the message index $i$ belongs to the demand's support, given that $\mathscr{C}_j$ is the code corresponding to the answer. 
This immediately implies that a necessary condition for individual privacy is that $q_{i,j}> 0$ for all $i,j$. 
This condition is, however, not sufficient. 
A necessary and sufficient condition for individual privacy is that for any $j\in [m]$, $q_{i,j}=q_j$ for all $i\in [K]$, for some $q_j>0$.
Also, a necessary and sufficient condition for recoverability is that for any $k,l,h$, $\mathscr{C}^{h}_{k,l}$ is $(k,l)$-feasible. 

For any $k,l$, let $d_{k,l}$ be the expected value of the dimension of a randomly chosen code from the ensemble $\{\mathscr{C}^{1}_{k,l},\dots,\mathscr{C}^{n}_{k,l}\}$ for the instance $(\mathrm{W}_k,\mathrm{V}_l)$, according to the probability distribution $\{p^{1}_{k,l},\dots,p^{n}_{k,l}\}$. 
Let $d_{\text{ave}}$ be the average of $d_{k,l}$'s over all $k,l$. 
Maximizing the rate of a randomized linear IPLT protocol, $1/d_{\text{ave}}$, is then equivalent to minimizing $d_{\text{ave}}$, subject to the necessary and sufficient conditions mentioned above for the individual privacy and recoverability conditions. 

\section{Proof of Theorem~\ref{thm:IPLT-UB}}\label{sec:IPLT-Conv}
In this section, we prove the result of Theorem~\ref{thm:IPLT-UB} by upper bounding the rate of IPLT protocols for any field size $q$. 

The proof relies on the following result which is a direct consequence of the individual privacy and recoverability conditions. 


\begin{lemma}\label{lem:NCIPLT}
Given any IPLT protocol, for any $i\in [K]$, there must exist $\tilde{\mathrm{W}}\in\mathbbm{W}$ with $i\in \tilde{\mathrm{W}}$, and $\tilde{\mathrm{V}}\in\mathbbm{V}$, such that \[H(\mathbf{Z}^{[\tilde{\mathrm{W}},\tilde{\mathrm{V}}]}| \mathbf{A}, \mathbf{Q}) = 0.\] 	
\end{lemma}

\begin{proof}
The proof is straightforward by the way of contradiction, and hence, omitted for brevity. 
\end{proof}


For (deterministic and randomized) linear IPLT protocols, the result of Lemma~\ref{lem:NCIPLT} is equivalent to the necessary (but not sufficient) conditions stated in Section~\ref{sec:LINEAR} for individual privacy.  
Notwithstanding that these necessary conditions are weaker than the necessary and sufficient conditions for individual privacy in Section~\ref{sec:LINEAR} for linear IPLT protocols, the former are less combinatorial and more information-theoretic. 
In addition, the necessary and sufficient conditions in Section~\ref{sec:LINEAR} are only applicable to linear protocols; whereas 
Lemma~\ref{lem:NCIPLT} applies to both linear and non-linear protocols. 

\begin{lemma}\label{lem:IPLT-Conv}
The rate of any IPLT protocol for $K$ messages, demand's support size $D$ and dimension $L$, is upper bounded by $(\lfloor {K}/{D}\rfloor+\min\{1,{R}/{L}\})^{-1}$. 
\end{lemma}

\begin{proof}
Consider an arbitrary IPLT protocol that generates the query-answer pair $(\mathrm{Q}^{[\mathrm{W},\mathrm{V}]},\mathrm{A}^{[\mathrm{W},\mathrm{V}]})$ for any given $\mathrm{W}$ and $\mathrm{V}$. 
For the ease of notation, we denote by $\mathbf{Q}$ and $\mathbf{A}$ the random variables $\mathbf{Q}^{[\mathbf{W},\mathbf{V}]}$ and $\mathbf{A}^{[\mathbf{W},\mathbf{V}]}$, respectively. 
To prove the upper bound on the rate, we need to show that ${H(\mathbf{A})\geq (L\lfloor {K}/{D}\rfloor+\min\{L,R\})B}$. 
Recall that 
${B=N\log_2 q}$ is the entropy of a uniformly distributed message over $\mathbbmss{F}_q^{N}$. 

Consider an arbitrary message index $k_1\in [K]$. 
By the result of Lemma~\ref{lem:NCIPLT}, there exist $\mathrm{W}_1\in \mathbbm{W}$ with $k_1\in \mathrm{W}_1$, and $\mathrm{V}_1\in\mathbbm{V}$ such that $H(\mathbf{Z}_1|\mathbf{A},\mathbf{Q})=0$, where $\mathbf{Z}_1\triangleq \mathbf{Z}^{[\mathrm{W}_1,\mathrm{V}_1]}$. 
By the same arguments as in the proof of \cite[Lemma~2]{HES2021JointJournal}, we have
\begin{align}
H(\mathbf{A})&\geq H(\mathbf{A}|\mathbf{Q})+H(\mathbf{Z}_1|\mathbf{A},\mathbf{Q}) \nonumber \\
& = {H(\mathbf{Z}_1|\mathbf{Q})+H(\mathbf{A}|\mathbf{Q},\mathbf{Z}_1}) \nonumber \\ 
& = {H(\mathbf{Z}_1)+H(\mathbf{A}|\mathbf{Q},\mathbf{Z}_1}). \label{eq:5}   
\end{align}
To further lower bound $H(\mathbf{A}|\mathbf{Q},\mathbf{Z}_1)$, we proceed as follows. Take an arbitrary message index $k_2\not\in \mathrm{W}_1$.
Again, by Lemma~\ref{lem:NCIPLT}, there exist $\mathrm{W}_2\in \mathbbm{W}$ with $k_2\in \mathrm{W}_2$, and $\mathrm{V}_2\in\mathbbm{V}$ such that $H(\mathbf{Z}_2|\mathbf{A},\mathbf{Q})=0$, where $\mathbf{Z}_2\triangleq \mathbf{Z}^{[\mathrm{W}_2,\mathrm{V}_2]}$. 
Using a similar technique as in~\eqref{eq:5}, it follows that $H(\mathbf{A}|\mathbf{Q},\mathbf{Z}_1)\geq  H(\mathbf{Z}_2|\mathbf{Q},\mathbf{Z}_1)+H(\mathbf{A}|\mathbf{Q},\mathbf{Z}_1,\mathbf{Z}_2)$, and consequently,  
\begin{equation}\label{eq:6}
H(\mathbf{A}|\mathbf{Q},\mathbf{Z}_1)\geq H(\mathbf{Z}_2|\mathbf{Z}_1)+H(\mathbf{A}|\mathbf{Q},\mathbf{Z}_2,\mathbf{Z}_1).    
\end{equation}
Combining~\eqref{eq:5} and~\eqref{eq:6}, we get 
\begin{equation}\label{eq:7}
H(\mathbf{A})\geq H(\mathbf{Z}_1)+H(\mathbf{Z}_2|\mathbf{Z}_1)+H(\mathbf{A}|\mathbf{Q},\mathbf{Z}_2,\mathbf{Z}_1).    
\end{equation} 



We repeat this lower-bounding process multiple rounds until there is no message index left to take. 
Let $n$ be the total number of rounds, and let $k_1,\dots,k_n$ be the $n$ message indices chosen over the rounds. 
For every $i\in [n]$, let $\mathrm{W}_i\in \mathbbm{W}$ with $k_i\in \mathrm{W}_i$ and $k_i\not\in \cup_{j=1}^{i-1}\mathrm{W}_j$, and $\mathrm{V}_i\in \mathbbm{V}$, be such that ${H(\mathbf{Z}_{i}|\mathbf{A},\mathbf{Q})=0}$, where $\mathbf{Z}_{i}\triangleq \mathbf{Z}^{[\mathrm{W}_i,\mathrm{V}_i]}$. 
(For any ${i\in [n]}$, the existence of $\mathrm{W}_i$ and $\mathrm{V}_i$ follows from the result of Lemma~\ref{lem:NCIPLT}.)
Note that $\cup_{i=1}^{n}\mathrm{W}_i = [K]$. 
This is because if $\cup_{i=1}^{n}\mathrm{W}_i \neq [K]$, the lower-bounding process could be continued for at least one more round (beyond $n$ rounds) by taking an arbitrary message index $k_{n+1}\in [K]\setminus \cup_{i=1}^{n}\mathrm{W}_i$, which contradicts with $n$ being the total number of rounds.  
Using the same technique as in~\eqref{eq:5} and~\eqref{eq:7}, we can show that
\begin{align}\label{eq:8}
H(\mathbf{A}) & \geq \sum_{i=1}^{n} H(\mathbf{Z}_i|\mathbf{Z}_{i-1},\dots,\mathbf{Z}_{1}) \nonumber \\ 
& \quad \quad 
+H(\mathbf{A}|\mathbf{Q},\mathbf{Z}_n,\dots,\mathbf{Z}_1)\nonumber\\
&\geq \sum_{i=1}^{n}  H(\mathbf{Z}_i|\mathbf{Z}_{i-1},\dots,\mathbf{Z}_{1}).
\end{align} 

Next, 
we show that 
\begin{equation}\label{eq:9}
H(\mathbf{Z}_i|\mathbf{Z}_{i-1},\dots,\mathbf{Z}_{1})\geq \min\{N_i,L\}B,    
\end{equation}
where ${N_i\triangleq |\mathrm{W}_i\setminus \cup_{j=1}^{i-1}\mathrm{W}_j|}$ is the number of message indices that belong to $\mathrm{W}_i$, but not $\cup_{j=1}^{i-1}\mathrm{W}_j$. 
(Note that ${N_1= |\mathrm{W}_1|=D}$.) 
Let $\mathbf{Z}_{i,1},\dots,\mathbf{Z}_{i,L}$ be the $L$ (row-) vectors pertaining to $\mathbf{Z}_i$, where $\mathbf{Z}_{i,l}\triangleq \mathrm{v}_{i,l} \mathbf{X}_{\mathrm{W}_i}$, and $\mathrm{v}_{i,l}$ is the $l$th row of the matrix $\mathrm{V}_i$ for each $l\in [L]$. 
The vectors $\mathbf{Z}_{i,1},\dots,\mathbf{Z}_{i,L}$ are linear combinations of the messages $\mathbf{X}_1,\dots,\mathbf{X}_K$. 
We need to show that there exist ${M_i\triangleq \min\{N_i,L\}}$ vectors pertaining to $\mathbf{Z}_i$ that are independent of all vectors pertaining to $\mathbf{Z}_1,\dots,\mathbf{Z}_{i-1}$.
Let $\mathrm{u}_{i,l}$ be a row-vector of length $K$ such that the vector $\mathrm{u}_{i,l}$ restricted to its components indexed by $\mathrm{W}_i$ is equal to the vector $\mathrm{v}_{i,l}$, and the rest of the components of the vector $\mathrm{u}_{i,l}$ are all zero, and 
let  
${\mathrm{U}_i \triangleq [\mathrm{u}^{\transpose}_{i,1},\dots,\mathrm{u}^{\transpose}_{i,L}]^{\transpose}}$. 
Using this notation, we need to show that the $L\times N$ matrix $\mathrm{U}_i$ contains $M_i$ rows that are linearly independent of the rows of the $L\times N$ matrices $\mathrm{U}_1,\dots,\mathrm{U}_{i-1}$. 
Note that the rows of $\mathrm{U}_i$ are linearly independent. 
This is because $\mathrm{U}_i$ contains $\mathrm{V}_i$ as a submatrix, and $\mathrm{V}_i$ has full rank (by assumption, $\mathrm{V}_i$ is MDS). 
Let $\mathrm{S}_i$ be an $L\times N_i$ submatrix of $\mathrm{U}_i$ formed by the columns indexed by $\mathrm{W}_i\setminus \cup_{j=1}^{i-1}\mathrm{W}_j$. 
Note that $\mathrm{S}_i$ is a submatrix of $\mathrm{V}_i$, and every $L\times L$ submatrix of $\mathrm{V}_i$ is invertible. 
Below, we consider two different cases: 
(i) $N_i\leq L$, and (ii) $N_i>L$.

In the case (i), the $N_i$ columns of $\mathrm{S}_i$ are linearly independent. 
Otherwise, any $L\times L$ submatrix of $\mathrm{V}_i$ that contains $\mathrm{S}_i$ cannot be invertible, and hence a contradiction. 
In the case (ii), any $L$ columns of $\mathrm{S}_i$ are linearly independent. 
Otherwise, $\mathrm{S}_i$ (and $\mathrm{V}_i$) contains an $L\times L$ submatrix that is not invertible, which is a contradiction. 
By these arguments, $\mathrm{rank}(\mathrm{S}_i)=M_i=\min\{L,N_i\}$, and hence, $\mathrm{S}_i$ contains $M_i$ linearly independent rows. 
Without loss of generality, assume that the first $M_i$ rows of $\mathrm{S}_i$ are linearly independent. 
Also, observe that the submatrix of $[\mathrm{U}^{\transpose}_{1},\dots,\mathrm{U}^{\transpose}_{i-1}]^{\transpose}$ restricted to its columns indexed by ${\mathrm{W}_i\setminus \cup_{j=1}^{i-1}\mathrm{W}_j}$ 
is an all-zero matrix. 
Thus, 
the first $M_i$ rows of $\mathrm{U}_i$ are linearly independent of the rows of $[\mathrm{U}^{\transpose}_{1},\dots,\mathrm{U}^{\transpose}_{i-1}]^{\transpose}$. 
This proves that there exist ${M_i}$ vectors pertaining to $\mathbf{Z}_i$ that are independent of all vectors pertaining to $\mathbf{Z}_1,\dots,\mathbf{Z}_{i-1}$. 
This completes the proof of~\eqref{eq:9}. 

Combining~\eqref{eq:8} and~\eqref{eq:9}, we have 
\begin{equation}\label{eq:10}
H(\mathbf{A})\geq \sum_{i=1}^{n} \min\{L,N_i\}B.    
\end{equation} 
Recall that $N_i = |\mathrm{W}_i\setminus \cup_{j=1}^{i-1} \mathrm{W}_j|$. 
Note that ${1\leq N_i\leq D}$ since $\mathrm{W}_i\setminus \cup_{j=1}^{i-1} \mathrm{W}_j$ is a subset of $\mathrm{W}_i$, and the message index $k_i$ belongs to ${\mathrm{W}_i\setminus \cup_{j=1}^{i-1} \mathrm{W}_j}$. 
Moreover, ${\sum_{i=1}^{n} N_i = K}$. 
This is because $\mathrm{W}_1$, ${\mathrm{W}_2\setminus \mathrm{W}_1}$, $\dots$, ${\mathrm{W}_n\setminus \cup_{j=1}^{n-1} \mathrm{W}_j}$ form a partition of $[K]$, and ${|\mathrm{W}_1|=N_1=D}$, ${|\mathrm{W}_2\setminus \mathrm{W}_1|=N_2}$, $\dots$, ${|\mathrm{W}_n\setminus \cup_{j=1}^{n-1} \mathrm{W}_j|=N_n}$. 

To obtain a converse bound, we need to minimize the right-hand side of~\eqref{eq:10}, namely, $\sum_{i=1}^{n} \min\{L,N_i\}$, subject to the constraints (i) ${N_1=D}$, and $1\leq N_i\leq D$ for any ${1<i\leq n}$, and (ii) ${\sum_{i=1}^{n} N_i = K}$. 
To solve this optimization problem, we first reformulate it using a change of variables as follows. 
For every ${j\in [D]}$, let ${T_j\triangleq \sum_{i=1}^{n} \mathbbm{1}_{\{N_i=j\}}}$ be the number of rounds $i$ such that ${N_i=j}$. 
Using this notation, the objective function $\sum_{i=1}^{n} \min\{L,N_i\}$ can be rewritten as $\sum_{j=1}^{D}  T_j\min\{L,j\}$, or equivalently, ${\sum_{j=1}^{L} T_j j+ \sum_{j=L+1}^{D} T_j L}$; 
the constraint (i) reduces to $T_j\in \mathbb{N}_0\triangleq\{0,1,\dots\}$ for every $1\leq j<D$, and ${T_D\in \mathbb{N}\triangleq \{1,2,\dots\}}$; 
and the constraint (ii) reduces to $\sum_{j=1}^{D} T_j j = K$. 
Thus, we need to solve the following integer linear programming (ILP) problem: 
\begin{eqnarray*}\label{eq:11}
& \hspace{-1cm} \mathrm{minimize} & \sum_{j=1}^{L} T_j j+ \sum_{j=L+1}^{D} T_j L,\\ \nonumber
& \hspace{-1cm} \mathrm{subject~to} & \sum_{j=1}^{D} T_j j= K,\\ \nonumber
&& T_1,\dots,T_{D-1}\in \mathbb{N}_0, T_D\in \mathbb{N}. 
\end{eqnarray*}
Solving this ILP using the Gomory's cutting-plane algorithm~\cite{MMWW2002}, it follows that an optimal solution is given by $T_D = \lfloor {K}/{D}\rfloor$, $T_R = 1$, and $T_j=0$ for all ${j\not\in\{R,D\}}$, where ${R=K \pmod D}$, and the optimal value of the objective function is given by ${L\lfloor {K}/{D} \rfloor+\min\{L,R\}}$. 
This implies that 
\begin{equation}\label{eq:11.5}
\sum_{i=1}^{n} \min\{L,N_i\}\geq L\left\lfloor\frac{K}{D} \right\rfloor+\min\{L,R\}.
\end{equation}
Combining~\eqref{eq:10} and~\eqref{eq:11.5}, ${H(\mathbf{A})\geq (L\lfloor {K}/{D} \rfloor+\min\{L,R\})B}$, as was to be shown.
\end{proof}

\section{Proof of Theorem~\ref{thm:IPLT-LB}}\label{sec:IPLT-Ach}
In this section, we present an IPLT protocol, termed the \emph{Generalized Partition-and-Code with Partial Interference Alignment (GPC-PIA) protocol}, which achieves the capacity lower bound of Theorem~\ref{thm:IPLT-LB} for sufficiently large field size $q$. 
In particular, when $L\leq S$, the GPC-PIA protocol is applicable for any $q\geq D+R$, and when $L>S$, the GPC-PIA protocol is applicable for any $q\geq D+R$, provided that the matrix $\mathrm{V}$ generates a Generalized Reed-Solomon (GRS) code~\cite{R2006}.
Examples of this protocol are provided in 
the appendix. 

With a slight abuse of notation, 
we denote by $\mathrm{W}$ (or $\overline{\mathrm{W}}$) a sequence of length $D$ (or $K-D$), instead of a set of size $D$ (or $K-D$), that is initially constructed by randomly permuting the $D$ message indices in the demand's support $\mathrm{W}$ (or the $K-D$ message indices in the complement of the demand's support ${[K]\setminus \mathrm{W}}$).  
Also, we denote by $\mathrm{V}$ an $L\times D$ matrix that is initially constructed by permuting the columns of the demand's coefficient matrix $\mathrm{V}$, according to the permutation used for constructing $\mathrm{W}$.

The GPC-PIA protocol consists of three steps as described below. 


\vspace{0.125cm}
\textbf{Step 1:} 
The user constructs a matrix $\mathrm{G}$ and a permutation $\pi$, and sends them as the query $\mathrm{Q}^{[\mathrm{W},\mathrm{V}]}$ to the server.
Depending on whether (i) $L\leq S$, or (ii) $L>S$, the construction of the matrix $\mathrm{G}$ and the permutation $\pi$ is different. 
We describe the construction for each of these two cases separately.

\subsection{Case (i)} 
In this case, $L\leq S$. 
Let $n\triangleq \lfloor {K}/{D}\rfloor-1$, $m\triangleq {R}/{S}+1$, and $t\triangleq {D}/{S}-1$. 
Note that $t+m = {(D+R)}/{S}$.

\subsubsection{Construction of the matrix $\mathrm{G}$}
The user constructs an $L(n+m)\times K$ matrix $\mathrm{G}$,
\begin{equation}\label{eq:12}
\mathrm{G} = \begin{bmatrix} 
\mathrm{G}_1 & 0 & \dots & 0 & 0 \\
0 & \mathrm{G}_2 & \dots & 0 & 0\\
\vdots & \vdots & \ddots & \vdots & \vdots \\
0 & 0 & \dots & \mathrm{G}_{n} & 0\\ 
0 & 0 & \dots & 0 & \mathrm{G}_{n+1} 
\end{bmatrix},
\end{equation} where the blocks $\mathrm{G}_1,\dots\mathrm{G}_{n}$ are $L\times D$ matrices, and the block $\mathrm{G}_{n+1}$ is an ${Lm\times (D+R)}$ matrix. 
The blocks $\mathrm{G}_1,\dots,\mathrm{G}_n,\mathrm{G}_{n+1}$ are constructed according to a randomized procedure as follows. 

The user randomly selects one of the blocks $\mathrm{G}_1,\dots,\mathrm{G}_{n+1}$, where the probability of selecting the block $\mathrm{G}_i$ for ${i\in [n]}$ is ${D}/{K}$, and the probability of selecting the block $\mathrm{G}_{n+1}$ is ${(D+R)}/{K}$. 
Let $b$ be the index of the selected block. 
Depending on whether $1\leq b\leq n$ or $b=n+1$, the description of the protocol is different.  

For the case of $1\leq b\leq n$, 
the user takes $\mathrm{G}_{b}=\mathrm{V}$, 
and takes $\mathrm{G}_i$ for each $i\in [n]\setminus \{b\}$ to be a randomly generated ${L\times D}$ MDS matrix.
The existence of such MDS matrices is guaranteed if the field size $q\geq D$. 
The construction of $\mathrm{G}_{n+1}$ is, however, different. 
First, the user randomly generates an $L\times (D+R)$ MDS matrix $\mathrm{C}$, and partitions the $D+R$ columns of $\mathrm{C}$ into $t+m$ ($={(D+R)}/{S}$) column-blocks each of size $L\times S$, i.e., $\mathrm{C} = [\mathrm{C}_1,\dots,\mathrm{C}_{t+m}]$, where $\mathrm{C}_i$ for $i\in [t+m]$ is an $L\times S$ matrix. 
(Such an MDS matrix $\mathrm{C}$ exists so long as the field size $q\geq D+R$.)
Then, the user constructs ${\mathrm{G}_{n+1} = [\mathrm{B}_{1},\mathrm{B}_{2}]}$, where $\mathrm{B}_{1}$ and $\mathrm{B}_{2}$ are given by
\begin{equation*}
\begin{bmatrix}
\alpha_1\omega_{1,1} \mathrm{C}_1 & \dots & \alpha_t\omega_{1,t} \mathrm{C}_t\\ 
\vdots & \vdots & \vdots \\ 
\alpha_1\omega_{m,1} \mathrm{C}_1 & \dots & \alpha_t\omega_{m,t} \mathrm{C}_t
\end{bmatrix} 
\end{equation*} and
\begin{equation*}
\begin{bmatrix}
\alpha_{t+1} \mathrm{C}_{t+1} &   &  \\ 
 & \ddots &  \\
 &  & \alpha_{t+m} \mathrm{C}_{t+m} \\
\end{bmatrix}
\end{equation*} 
respectively. 
Here, the parameters $\alpha_1,\dots,\alpha_{t+m}$ are $t+m$ randomly chosen elements from $\mathbbmss{F}_q\setminus \{0\}$, and the parameters $\omega_{i,j}\triangleq (x_i-y_j)^{-1}$ for $i\in [m]$ and $j\in [t]$, where $x_1,\dots,x_m$ and $y_1,\dots,y_t$ are $t+m$ distinct elements chosen at random from $\mathbbmss{F}_q$. 
Note that $\omega_{i,j}$ is the entry $(i,j)$ of an $m\times t$ Cauchy matrix.


Now, consider the case of $b=n+1$. 
For each ${i\in [n]}$, the user takes $\mathrm{G}_i$ to be a randomly generated $L\times D$ MDS matrix, and constructs $\mathrm{G}_{n+1}$ with a structure similar to that in the previous case, 
but the column-blocks $\mathrm{C}_1,\dots,\mathrm{C}_{t+m}$ and the parameters $\alpha_{1},\dots,\alpha_{t+m}$ are chosen differently.

\subsubsection*{Construction of column-blocks $\mathrm{C}_1,\dots,\mathrm{C}_{t+m}$} To construct the column-blocks $\mathrm{C}_i$'s, the user proceeds as follows. 
\begin{itemize}
    \item First, the user partitions the $D$ columns of $\mathrm{V}$ into ${t+1}$ ($=\lfloor {K}/{D}\rfloor$) column-blocks each of size $L\times S$, i.e., 
${{\mathrm{V}}= [{\mathrm{V}}_1,\dots,{\mathrm{V}}_{t+1}]}$, 
where $\mathrm{V}_i$ for $i\in [t+1]$ is an $L\times S$ matrix. 
\item The user then randomly chooses $t+1$ indices from $[t+m]$, say, $r$ indices $k_1,\dots,k_r \in [t]$ and $s$ indices $l_1,\dots,l_s \in [t+1:t+m]$ such that $r+s=t+1$. 
Note that the column-blocks indexed by $k_1,\dots,k_s$ belong to the matrix $\mathrm{B}_1$, and the column-blocks indexed by $l_1,\dots,l_s$ belong to the matrix $\mathrm{B}_2$.  
\item Then, the user takes $\mathrm{C}_{k_j}={\mathrm{V}}_j$ for $j\in [r]$, and $\mathrm{C}_{l_j}={\mathrm{V}}_{r+j}$ for $j\in [s]$.  
\item The user then randomly generates the rest of $\mathrm{C}_i$'s for ${i\not\in \{k_1,\dots,k_r,l_1,\dots,l_{s}\}}$ such that the matrix $\mathrm{C} = [\mathrm{C}_1,\dots,\mathrm{C}_{t+m}]$ is an MDS matrix.
\end{itemize}


%
\subsubsection*{Choice of parameters $\alpha_1,\dots,\alpha_{t+m}$} Before explaining the process of choosing the parameters $\alpha_i$'s, we introduce a few more definitions and notations. 

We refer to the $L\times (D+R)$ submatrix of $\mathrm{G}_{n+1}$ formed by the $i$th block of $L$ rows as the $i$th row-block of $\mathrm{G}_{n+1}$. 
Note that $\mathrm{G}_{n+1}$ has $m$ row-blocks.


Note that $\{k_1,\dots,k_{r},l_1,\dots,l_s\}$ is the index set of those column-blocks of $\mathrm{C}$ that correspond to the column-blocks of ${\mathrm{V}}$. 
Note also that every $\mathrm{C}_{i}$ for $i\in \{k_1,\dots,k_r\}$ appears in all row-blocks of $\mathrm{G}_{n+1}$, and 
every $\mathrm{C}_{i}$ for $i\in \{l_1,\dots,l_s\}$ appears only in the $(i-t)$th row-block of $\mathrm{G}_{n+1}$.

We define $\{k_{r+1},\dots,k_t\}\triangleq[t]\setminus \{k_1,\dots,k_r\}$ as the index set of those column-blocks of $\mathrm{C}$ belonging to the matrix $\mathrm{B}_1$ that do not correspond to any column-blocks of $\mathrm{V}$, and $\{l_{s+1},\dots,l_m\}\triangleq[t+1:m]\setminus \{l_1,\dots,l_s\}$ as the index set of those column-blocks of $\mathrm{C}$ belonging to the matrix $\mathrm{B}_2$ that do not correspond to any column-blocks of $\mathrm{V}$.

The parameters $\alpha_i$'s are to be chosen such that, by performing row-block operations on $\mathrm{G}_{n+1}$, the user can construct an $L\times (D+R)$ matrix---composed of $t+m$ column-blocks, each of size $L\times S$---that satisfies the following two conditions: 
\begin{itemize}
    \item[(a)] The column-blocks indexed by $\{k_{r+1},\dots,k_t\}$ and $\{l_{s+1},\dots,l_{m}\}$ are all-zero;
    \item[(b)] The column-blocks indexed by $\{k_{1},\dots,k_r\}$ are $\mathrm{C}_{k_1},\dots,\mathrm{C}_{k_r}$, and the column-blocks indexed by $\{l_{1},\dots,l_{s}\}$ are $\mathrm{C}_{l_{1}},\dots,\mathrm{C}_{l_s}$.
\end{itemize}


To perform row-block operations on $\mathrm{G}_{n+1}$, 
the user multiplies the $(i-t)$th row-block of $\mathrm{G}_{n+1}$ by a nonzero coefficient $c_{i}$ for ${i\in\{l_1,\dots,l_s\}}$. 
Let $\mathrm{c} \triangleq [c_{l_1},\dots,c_{l_s}]^{\transpose}$. 

Followed by choosing $\alpha_{k_{r+1}},\dots,\alpha_{k_{t}}$ randomly from ${\mathbbmss{F}_q\setminus \{0\}}$, 
it is easy to verify that the condition (a) is met so long as the vector $\mathrm{M}_1\mathrm{c}$ is all-zero, where  
\begin{equation*}
\mathrm{M}_1 \triangleq 
\begin{bmatrix}
\omega_{l_1-t,k_{r+1}} & \omega_{l_2-t,k_{r+1}} & \dots & \omega_{l_s-t,k_{r+1}}\\
\vdots & \vdots & \vdots & \vdots \\
\omega_{l_1-t,k_t} & \omega_{l_2-t,k_t} & \dots & \omega_{l_s-t,k_t}
\end{bmatrix}.
\end{equation*} Since $\mathrm{M}_1$ is a Cauchy matrix by the choice of $\omega_{i,j}$'s, every $(s-1)\times (s-1)$ submatrix of $\mathrm{M}_1$ is invertible~\cite{R2006}. 
This implies that, for any arbitrary ${c_{l_1}\neq 0}$, there is a unique solution for the vector $\mathrm{c}$ such that $\mathrm{M}_1\mathrm{c}$ is all-zero, and the vector $\mathrm{c}$ does not contain any zeros. 
Given the vector $\mathrm{c}$, 
it is easy to see that the condition (b) is met so long as ${\alpha_{l_1} = 1/c_{l_1}}, \dots, {\alpha_{l_s} = 1/c_{l_s}}$, and 
$\alpha_{k_1},\dots,\alpha_{k_{r}}$ are such that the vector $\mathrm{M}_2\mathrm{c}$ is all-one, where  
\begin{equation*}
\mathrm{M}_2 \triangleq 
\begin{bmatrix}
\alpha_{k_1}\omega_{l_1-t,k_1} & \dots & \alpha_{k_1}\omega_{l_s-t,k_1}\\
\vdots & \vdots & \vdots \\
\alpha_{k_r}\omega_{l_1-t,k_r} & \dots & \alpha_{k_r}\omega_{l_s-t,k_r}
\end{bmatrix}.
\end{equation*} 
Solving 
for the variables $\alpha_{k_1},\dots,\alpha_{k_r}$, it follows that \[\alpha_{k_j} = \left(\sum_{i=1}^{s} c_{l_i}\omega_{l_i-t,k_j}\right)^{-1}\] for ${j\in [r]}$. 
Note that $\alpha_{k_1},\dots,\alpha_{k_r}$ are nonzero, and $\sum_{i=1}^{s} c_{l_i}\omega_{l_i-t,k_j}$ is nonzero. 
This can be easily shown as follows. 
Let $\mathrm{M}$ be a matrix formed by vertically concatenating $\mathrm{M}_1$ and the $j$th row of $\mathrm{M}_2$ normalized by $\alpha_{k_j}$. 
Note that the first $s-1$ components of the vector $\mathrm{M}\mathrm{c}$ are all zero because $\mathrm{M}_1\mathrm{c}$ is all-zero, 
and the last component of $\mathrm{M}\mathrm{c}$ is $\sum_{i=1}^{s} c_{l_i}\omega_{l_i-t,k_j}$. 
If $\sum_{i=1}^{s} c_{l_i}\omega_{l_i-t,k_j}$ is zero, then $\mathrm{M}\mathrm{c}$ is all-zero.
Since the vector $\mathrm{c}$ is not all-zero, then the rows of $\mathrm{M}$ must be linearly dependent.
This is, however, a contradiction because $\mathrm{M}$ is a Cauchy matrix, and hence, the rows of $\mathrm{M}$ are linearly independent.
Thus, $\sum_{i=1}^{s} c_{l_i}\omega_{l_i-t,k_j}$ is nonzero. 

Lastly, the user chooses $\alpha_{l_{s+1}},\dots,\alpha_{l_{m}}$ randomly from ${\mathbbmss{F}_q\setminus \{0\}}$. 
This concludes the process of choosing the parameters $\alpha_1,\dots,\alpha_{t+m}$. 



\subsubsection{Construction of the permutation $\pi$}
For the ease of notation, suppose ${\mathrm{W}=\{i_1,\dots,i_{D}\}}$ and ${\overline{\mathrm{W}}={\{i_{D+1},\dots,i_{K}\}}}$. 

First, consider the case of $1\leq b\leq n$. 
The user constructs the permutation $\pi$ as follows: $\pi(i_j) = (b-1)D+j$ for ${j\in [D]}$, and $\pi(i_j)$ for ${j\in [D+1:K]}$ is randomly chosen from ${[K]\setminus \{\pi(i_k): k\in [j-1]\}}$. 

Next, consider the case of $b = n+1$. 
Recall that $k_1,\dots,k_{r},l_1,\dots,l_{s}$ are the indices of the column-blocks of $\mathrm{C}$ that correspond to the column-blocks of $\mathrm{V}$.  
Let ${e_j\triangleq \lceil {j}/{S}\rceil}$ for ${j\in [rS]}$, and ${e_j\triangleq \lceil {j}/{S}\rceil-r}$ for ${j\in [rS+1:D]}$, and ${f_j \triangleq S}$ if ${S\mid j}$, and ${f_j\triangleq j\pmod S}$ if ${S\nmid j}$. 
The user constructs the permutation $\pi$ as follows: 
${\pi(i_j) = nD+(k_{e_j}-1)S+f_j}$ for ${j\in [rS]}$, ${\pi(i_j) = nD+(l_{e_j}-1)S+f_j}$ for ${j\in [rS+1:D]}$, and $\pi(i_j)$ for ${j\in [D+1:K]}$ is randomly chosen from ${[K]\setminus \{\pi(i_k): k\in [j-1]\}}$. 

\subsection{{Case (ii)}} 
Recall that in this case, $L>S$. 
Let $n\triangleq \lfloor {K}/{D}\rfloor-1$, and ${m\triangleq {R}/{L}+1}$. 
Note that here $n$ is defined the same as in the case~(i), but $m$ is defined differently. 

\subsubsection{Construction of the matrix $\mathrm{G}$}
The user constructs an $L(n+m)\times K$ matrix $\mathrm{G}$ with a structure similar to~\eqref{eq:12}, where $\mathrm{G}_1,\dots,\mathrm{G}_n$ are constructed similarly as in the case (i), but the construction of $\mathrm{G}_{n+1}$ is different. 
Below, we explain how $\mathrm{G}_{n+1}$ is constructed in this case. 

For the case of $1\leq b\leq n$, the user takes $\mathrm{G}_{n+1}$ to be a randomly generated $(L+R)\times (D+R)$ MDS matrix. 
Such an MDS matrix exists so long as the field size $q\geq D+R$. 

For the case of $b=n+1$, the user constructs $\mathrm{G}_{n+1}$ using a similar technique as in the JPLT protocol of~\cite{HES2021JointJournal}. 
First, the user randomly chooses $D$ indices from $[D+R]$, say, $h_1,\dots,h_D$. 
The user then constructs a $(D-L)\times D$ parity-check matrix $\myLambda$ of the $[D,L]$ MDS code generated by $\mathrm{V}$. 
Then, the user constructs a ${(D-L)\times (D+R)}$ MDS matrix $\mathrm{H}$ such that $\myLambda$ is a submatrix of $\mathrm{H}$ formed by the columns indexed by $\{h_1,\dots,h_D\}$. 
The user then takes $\mathrm{G}_{n+1}$ to be an $(L+R)\times (D+R)$ generator matrix of the $[D+R,L+R]$ MDS code defined by the parity-check matrix $\mathrm{H}$. 

The existence of such a matrix $\mathrm{H}$---that satisfies the above conditions, depends in general on $D,L,R$, the field size $q$, and the structure of the matrix $\myLambda$ (or $\mathrm{V}$). Using Schwartz–Zippel lemma, it can be shown that such a matrix always exists when $q$ is sufficiently large. 
In addition, such a matrix can be constructed systematically for any $q\geq D+R$ when $\myLambda$ (or $\mathrm{V}$) is a Vandermonde matrix with distinct parameters (or more generally, the product of a Vandermonde matrix with distinct parameters and a diagonal matrix with nonzero entries on the main diagonal)~\cite{R2006}. 



\subsubsection{Construction of the permutation $\pi$}
Similarly as before, suppose ${\mathrm{W}}=\{i_1,\dots,i_{D}\}$ and $\overline{\mathrm{W}}=\{i_{D+1},\dots,i_{K}\}$.
For the case of ${1\leq b\leq n}$, the permutation $\pi$ is constructed the same as in the case (i), whereas, for the case of ${b=n+1}$, the construction is different from that in the case (i). 
In this case, the user constructs $\pi$ as follows: $\pi(i_j) = nD+h_{j}$ for ${j\in [D]}$, and $\pi(i_j)$ is randomly chosen from $[K]\setminus \{\pi(i_k): k\in [j-1]\}$ for ${j\in [D+1:K]}$.

\vspace{0.125cm} 
\textbf{Step 2:} Given the query $\mathrm{Q}^{[\mathrm{W},\mathrm{V}]}$, i.e., the matrix $\mathrm{G}$ and the permutation $\pi$, the server first constructs the matrix $\tilde{\mathrm{X}} \triangleq \pi(\mathrm{X})$ by permuting the rows of the matrix $\mathrm{X}$ according to the permutation $\pi$, i.e., for every $i\in [K]$, $\pi(i)$th row of $\tilde{\mathrm{X}}$ is the $i$th row of $\mathrm{X}$. 
Then, the server  
computes the matrix $\mathrm{Y}\triangleq \mathrm{G}\tilde{\mathrm{X}}$, and sends $\mathrm{Y}$ back to the user as the answer $\mathrm{A}^{[\mathrm{W},\mathrm{V}]}$.\vspace{0.125cm} 

\textbf{Step 3:} Upon receiving the answer $\mathrm{A}^{[\mathrm{W},\mathrm{V}]}$, i.e., the matrix $\mathrm{Y}$, the user recovers the demand matrix $\mathrm{Z}^{[\mathrm{W},\mathrm{V}]}$ as follows. 
Let $\mathrm{Y}_i$ for $i\in [n]$ be a submatrix of $\mathrm{Y}$ formed by the rows indexed by $[(i-1)L+1:iL]$, and let $\mathrm{Y}_{n+1}$ be a submatrix of $\mathrm{Y}$ formed by the rows indexed by $[nL+1:(n+m)L]$. 
For the case of $1\leq b\leq n$,  $\mathrm{Z}^{[\mathrm{W},\mathrm{V}]}$ can be recovered from the matrix $\mathrm{Y}_{b}$ for both cases (i) and (ii). 
For the case of ${b=n+1}$, $\mathrm{Z}^{[\mathrm{W},\mathrm{V}]}$ can be recovered 
by performing proper row-block or row operations on the augmented matrix $[\mathrm{G}_{n+1},\mathrm{Y}_{n+1}]$ for the case (i) or (ii), respectively.

\begin{lemma}\label{lem:IPLT-Ach}
The GPC-PIA protocol is an IPLT protocol, and achieves the rate ${(\lfloor {K}/{D}\rfloor+\min\{{R}/{S},{R}/{L}\})^{-1}}$. 
\end{lemma}

\begin{proof}
To avoid repetition, we only present the proof for the case (i). 
Using the same arguments, the results can be shown for the case (ii).  

In the case (i), it is easy to see that the rate of the protocol is ${LB/(L(n+m)B)} = {(n+m)^{-1}} = (\lfloor {K}/{D}\rfloor+{R}/{S})^{-1}$. 
This is because the matrix $\mathrm{G}$ has $L(n+m)$ rows, and the matrix $\mathrm{Y}=\mathrm{G}\tilde{\mathrm{X}}$ contains $L(n+m)$ independently and uniformly distributed row-vectors of length $N$ with entries from $\mathbbmss{F}_q$, each with entropy $B=N\log_2 q$. 

The proof of recoverability is as follows. 
For the case of $1\leq b\leq n$, it is straightforward to see that $\mathrm{Y}_b = \begin{bmatrix}0_{L\times (b-1)D} & \mathrm{G}_b & 0_{L\times (K-bD)}\end{bmatrix}\tilde{\mathrm{X}} = \mathrm{G}_b\tilde{\mathrm{X}}_{[(b-1)D+1:bD]} = \mathrm{V}{\mathrm{X}}_{\mathrm{W}}=\mathrm{Z}^{[\mathrm{W},\mathrm{V}]}$. 
This is because $\mathrm{G}_b = \mathrm{V}$ by Step~1 of the protocol, $\pi(\mathrm{W})=[(b-1)D+1:bD]$ by the construction of the permutation $\pi$ in Step~1 of the protocol, and $\tilde{\mathrm{X}}_{\pi(\mathrm{W})} = \mathrm{X}_{\mathrm{W}}$ by Step~2 of the protocol.
Now, consider the case of $b=n+1$. 
Recall that the row-block operations on $\mathrm{G}_{n+1}$ are performed on the row-blocks indexed by $\{l_1-t,\dots,l_s-t\}$. 
Recall also that the vector $\mathrm{c}=[c_{l_1},\dots,c_{l_s}]^{\transpose}$ defined in Step~1 of the protocol represents the coefficients required for performing these row-block operations.
Let $\tilde{\mathrm{G}}_{n+1} = [\tilde{\mathrm{B}}_1,\tilde{\mathrm{B}}_2]$ be a submatrix of $\mathrm{G}_{n+1} = [\mathrm{B}_1,\mathrm{B}_2]$ formed by the row-blocks indexed by $\{l_1-t,\dots,l_s-t\}$. 
Note that $\tilde{\mathrm{B}}_1$ and $\tilde{\mathrm{B}}_2$ are given by
\begin{equation*}
\setlength\arraycolsep{4pt}
\begin{bmatrix}
\alpha_1\omega_{l_1-t,1} \mathrm{C}_1 & \dots & \alpha_t\omega_{l_1-t,t} \mathrm{C}_t\\ 
\vdots & \vdots & \vdots \\ 
\alpha_1\omega_{l_s-t,1} \mathrm{C}_1 & \dots & \alpha_t\omega_{l_s-t,t} \mathrm{C}_t
\end{bmatrix},
\end{equation*} and\vspace{0.65cm} 
\begin{equation*}
\setlength\arraycolsep{5pt}
\begin{bmatrix}
\bovermat{$l_1-1$}{\hspace{0.25cm}0} & \alpha_{l_1}\mathrm{C}_{l_1} & \hspace{-0.5cm} \bovermat{$l_2-l_1-1$}{\hspace{0.5cm}0} &  &  &  & & & \\ 
& & 0 & \alpha_{l_2}\mathrm{C}_{l_2} & 0 &  & & \\ 
& & & & & \ddots & & & \\
&  &  & & & & \hspace{-0.5cm}\bundermat{$l_{s}-l_{s-1}-1$}{\hspace{0.625cm} 0} & \hspace{0.125cm}\alpha_{l_s}\mathrm{C}_{l_s} & \hspace{-0.35cm} \bundermat{$t+m-l_s$}{\hspace{0.475cm}0} & 
\end{bmatrix},\vspace{0.65cm}
\end{equation*} respectively, where $0$'s are $L\times S$ all-zero matrices. 
Thus, multiplying the $s$ row-blocks of the matrix $\tilde{\mathrm{G}}_{n+1}$ by the components of the vector $\mathrm{c}$, namely, $c_{l_1},\dots,c_{l_s}$, and
summing the row-blocks of the resulting matrix, it follows that: 
(i) the column-blocks indexed by $\{k_1,\dots,k_r\}$ are given by $\mathrm{C}_{k_1},\dots,\mathrm{C}_{k_r}$, or equivalently, $\mathrm{V}_{1},\dots,\mathrm{V}_{r}$, because $\sum_{i=1}^{s} \alpha_{k_j}\omega_{l_i-t,k_j}=1$ by the choice of $\alpha_{k_j}$ for $j\in [r]$ in Step~1 of the protocol; 
(ii) the column-blocks indexed by $\{k_{r+1},\dots,k_t\}$ are all zero, because for $j\in [r+1:t]$, $\sum_{i=1}^{s} \alpha_{k_j}\omega_{l_i-t,k_j}=\alpha_{k_j}\sum_{i=1}^{s} \omega_{l_i-t,k_j}$, and $\sum_{i=1}^{s} \omega_{l_i-t,k_j}$ is the $(j-r)$th component of the vector $\mathrm{M}_1\mathrm{c}$, which is itself an all-zero vector, as discussed in Step~1 of the protocol; 
(iii) the column-blocks indexed by $\{l_1,\dots,l_s\}$ are given by $\mathrm{C}_{l_1},\dots,\mathrm{C}_{l_s}$, or equivalently, $\mathrm{V}_{r+1},\dots,\mathrm{V}_{t+1}$, because $c_{l_j}\alpha_{l_j}=1$ for $j\in [s]$ by the choice of $\alpha_{l_j}$ for $j\in [s]$ in Step~1 of the protocol;
and (iv) the column-blocks indexed by $\{l_{s+1},\dots,l_m\}$ are all-zero matrices. 
Thus, by performing these row-block operations on $\tilde{\mathrm{G}}_{n+1}$, the user obtains a single row-block that contains $t+m$ column-blocks, each of size $L\times S$, where the $t+1$ columns-blocks indexed by $\{k_1,\dots,k_r\}\cup \{l_1,\dots,l_s\}$ form the matrix $\mathrm{V}$, and the rest of the column-blocks are all-zero matrices. 
Let $\tilde{\mathrm{W}}\triangleq \cup_{i\in \{k_1,\dots,k_r\}\cup \{l_1,\dots,l_s\}} \tilde{\mathrm{W}}_i$, where ${\tilde{\mathrm{W}}_i\triangleq [nD+(i-1)S+1:nD+iS]}$. 
Note that ${\tilde{\mathrm{X}}_{\tilde{\mathrm{W}}} = \mathrm{X}_{\mathrm{W}}}$. 
This is because $\tilde{\mathrm{W}}=\pi(\mathrm{W})$ by the construction of the permutation $\pi$ in Step~1 of the protocol.  
Thus, the user can perform these row-block operations on $\tilde{\mathrm{Y}}_{n+1}\triangleq \tilde{\mathrm{G}}_{n+1}\tilde{\mathrm{X}}$, and recover the demand matrix $\mathrm{V}\tilde{\mathrm{X}}_{\tilde{\mathrm{W}}}=\mathrm{V}\mathrm{X}_{\mathrm{W}}$. 
This completes the proof of recoverability.

Next, we show that the individual privacy condition is satisfied. 
Let $\tilde{\mathrm{X}} \triangleq [X_{i_1}^{\transpose},\dots,X_{i_K}^{\transpose}]^{\transpose}$. 
For each ${j\in [n]}$, let $\mathrm{I}_j$ be the set of $j$th group of $D$ elements in $\{i_1,\dots,i_{nD}\}$, and for each $j\in [t+m]$, let $\mathrm{I}_{n+j}$ be the set of $j$th group of $S$ elements in $\{i_{nD+1},\dots,i_K\}$. 
Let $T_1\triangleq C_{t+m,t+1}$. 
For each $j\in [n]$, let $\mathrm{W}_j\triangleq \mathrm{I}_j$, and for each $j\in [T_1]$, let $\mathrm{W}_{n+j} = \cup_{k\in \mathrm{J}_j} \mathrm{I}_{k}$, where $\mathrm{J}_1,\dots,\mathrm{J}_{T_1}$ are all $(t+1)$-subsets of ${[n+1:n+t+m]}$.
It is easy to verify that $\mathrm{W}_1,\dots,\mathrm{W}_n,\mathrm{W}_{n+1},\dots,\mathrm{W}_{n+t+m}$ are the only possible demand's supports, from the server's perspective, given the user's query. 

Let ${\mathrm{Q} \triangleq \{\mathrm{G},\pi\}}$ be the user's query. 
To prove that the individual privacy condition is satisfied, we need to show that ${\Pr(i\in \mathbf{W}|\mathbf{Q}=\mathrm{Q})} = {\Pr(i\in \mathbf{W})}={D}/{K}$ for all $i\in [K]$. 
Fix an arbitrary $i\in [K]$. In the following, we consider two different cases: (i) ${\pi(i)\leq nD}$, and (ii) ${\pi(i)>nD}$. 

First, consider the case (i). 
In this case, there exists a unique $j\in [n]$ such that $i\in \mathrm{W}_j$. Thus, \[{\Pr(i\in \mathbf{W}|\mathbf{Q}=\mathrm{Q})} = {\Pr(\mathbf{W}=\mathrm{W}_j|\mathbf{Q}=\mathrm{Q})}.\] 
By applying Bayes' rule, we have
\begin{align}
& \Pr(\mathbf{W}=\mathrm{W}_j|\mathbf{Q}=\mathrm{Q})\nonumber\\
& = \frac{\Pr(\mathbf{Q}=\mathrm{Q}|\mathbf{W}=\mathrm{W}_j)}{\Pr(\mathbf{Q}=\mathrm{Q})}\Pr(\mathbf{W}=\mathrm{W}_j). \label{eq:13} 
\end{align}
Recall that $\Pr(\mathbf{W}=\mathrm{W}_j)={1}/{C_{K,D}}$. 
By the construction, the structure of $\mathrm{G}$, i.e., the size and the position of the blocks $\mathrm{G}_1,\dots,\mathrm{G}_{n+1}$, does not depend on $(\mathrm{W},{\pi})$, and the matrix $\mathrm{V}$ and all other MDS matrices used in the construction of $\mathrm{G}$ are generated independently from $(\mathrm{W},{\pi})$. 
Thus, $\mathbf{G}$ is independent of $(\mathbf{W},\boldsymbol{\pi})$. 
Obviously, $\Pr(\mathbf{Q}=\mathrm{Q}) = {\Pr(\mathbf{G}=\mathrm{G},\boldsymbol{\pi}=\pi)}$. 
Then, we can write 
\begin{align}
& \frac{\Pr(\mathbf{Q}=\mathrm{Q}|\mathbf{W}=\mathrm{W}_j)}{\Pr(\mathbf{Q}=\mathrm{Q})}\nonumber\\
& = \frac{\Pr(\mathbf{G} = \mathrm{G})\Pr(\boldsymbol{\pi}=\pi|\mathbf{W}=\mathrm{W}_j)}{\Pr(\mathbf{G}=\mathrm{G})\Pr(\boldsymbol{\pi}=\pi)} \nonumber \\ 
& = \frac{\Pr(\boldsymbol{\pi}=\pi|\mathbf{W}=\mathrm{W}_j)}{\Pr(\boldsymbol{\pi}=\pi)}. \label{eq:14}
\end{align} 
Obviously, $\Pr(\boldsymbol{\pi}=\pi) = {1}/{K!}$.
Given $\mathbf{W} = \mathrm{W}_j$, the conditional probability of the event of $\boldsymbol{\pi}=\pi$ is equal to the joint probability of the two events ${\boldsymbol{\pi}(\mathbf{W}) = \pi(\mathrm{W}_j)}$ and ${\boldsymbol{\pi}([K]\setminus \mathbf{W}) = \pi([K]\setminus\mathrm{W}_j)}$.
Let $\mathbf{b}$ be a random variable representing the index of the block selected by the user in Step~1 of the protocol. 
Then, we have 
\begin{align}\label{eq:12.1}
{\Pr(\boldsymbol{\pi}(\mathbf{W}) = \pi(\mathrm{W}_j))} & = \Pr(\mathbf{b} = j)\times \frac{1}{D!}=\frac{D}{K}\times \frac{1}{D!}.    
\end{align}
In addition, by the construction of $\pi$ as in Step~1 of the protocol, we have 
\begin{align}\label{eq:12.2}
\Pr(\boldsymbol{\pi}([K]\setminus \mathbf{W}) = \pi([K]\setminus \mathrm{W}_j)) = \frac{1}{(K-D)!}.
\end{align} 
By~\eqref{eq:12.1} and~\eqref{eq:12.2}, we have
\begin{equation}\label{eq:4.12}
\Pr(\boldsymbol{\pi}=\pi|\mathbf{W}=\mathrm{W}_j) = \frac{D}{K}\times\frac{1}{D!}\times \frac{1}{(K-D)!}.    
\end{equation}
Combining~\eqref{eq:13}-\eqref{eq:4.12}, we have
\begin{equation}\label{eq:4.13}
\Pr(i\in \mathbf{W}|\mathbf{Q}=\mathrm{Q}) =
K!\times\frac{D}{K}\times\frac{1}{K!} =\frac{D}{K}.
\end{equation}

Now, consider the case (ii).
Let $T_2\triangleq C_{t+m-1,t}$.
Note that $T_2/T_1 = D/(D+R)$. 
One can easily verify that there exist distinct indices $j_1,\dots,j_{T_2}\in [T_1]$ 
such that ${i\in \mathrm{W}_{n+j_k}}$ for all $k\in [T_2]$. 
Thus, we can write
\begin{align}
& {\Pr(i\in \mathbf{W}|\mathbf{Q}=\mathrm{Q})} \nonumber\\
& = {\sum_{k=1}^{T_{2}}\Pr(\mathbf{W}=\mathrm{W}_{n+j_k}|\mathbf{Q}=\mathrm{Q})}\nonumber\\
&  = \sum_{k=1}^{T_{2}} \frac{\Pr(\mathbf{Q}=\mathrm{Q}|\mathbf{W}=\mathrm{W}_{n+j_k})}{\Pr(\mathbf{Q}=\mathrm{Q})}\Pr(\mathbf{W}=\mathrm{W}_{n+j_k}) \nonumber\\
&  = \sum_{k=1}^{T_{2}} \frac{\Pr(\mathbf{G} = \mathrm{G})\Pr(\boldsymbol{\pi}=\pi|\mathbf{W}=\mathrm{W}_{n+j_k})}{\Pr(\mathbf{G}=\mathrm{G})\Pr(\boldsymbol{\pi}=\pi)}\times\frac{1}{C_{K,D}}\nonumber\\
&  = \sum_{k=1}^{T_{2}} \frac{\Pr(\boldsymbol{\pi}=\pi|\mathbf{W}=\mathrm{W}_{n+j_k})}{\Pr(\boldsymbol{\pi}=\pi)}\times\frac{1}{C_{K,D}}\nonumber\\
&  = \sum_{k=1}^{T_{2}} \left(K!\times\frac{D+R}{K}\times \frac{1}{T_{1}}\times\frac{1}{D!}\times\frac{1}{(K-D)!}\times \frac{1}{C_{K,D}}\right)\nonumber\\
&  = T_{2} \left(K!\times\frac{D+R}{K}\times \frac{1}{T_{1}}\times\frac{1}{D!}\times\frac{1}{(K-D)!}\times \frac{1}{C_{K,D}}\right)\nonumber\\
& = K!\times \frac{T_{2}}{T_{1}}\times \frac{D+R}{K}\times \frac{1}{K!} = \frac{D}{D+R}\times \frac{D+R}{K} = \frac{D}{K}.\label{eq:4.14}
\end{align} 
By~\eqref{eq:4.13} and~\eqref{eq:4.14}, we have ${\Pr(i\in \mathbf{W}|\mathbf{Q}=\mathrm{Q})}={D}/{K}$ for all $i\in [K]$. This completes the proof of individual privacy. 
\end{proof}

\section{Conclusion and Future Work}\label{sec:Con}
In this work, we considered the problem of single-server Private Linear Transformation (PLT) with individual privacy guarantees (or IPLT).
This problem includes 
a single remote server that stores a dataset of $K$ messages, and a user that wishes to compute $L$ linear combinations of a $D$-subset of the messages. 
The goal is to perform the computation by downloading the minimum possible amount of information from the server, while keeping the identity of every individual message required for the user's computation private. 
The IPLT problem generalizes the problems of single-server Private Information Retrieval (PIR) with individual privacy (or IPIR) and single-server Private Linear Computation (PLC) with individual privacy (or IPLC).

We focused on the setting in which the coefficient matrix of the required linear combinations is a maximum distance separable (MDS) matrix. 
For this setting, we established lower and upper bounds on the capacity of IPLT, where the capacity is defined as the supremum of all achievable download rates. 
We also showed that our bounds are tight under certain conditions. 
Comparing our results with those for the problem of single-server PLT under the stricter notion of joint privacy, we showed that IPLT can be performed more efficiently than PLT with joint privacy, in terms of the download cost, for a wide range of problem parameters.   




Several problems---closely related to the IPLT problem---are left open. 
Below, we list a few of these problems. 
\begin{itemize}
    \item[1)] The capacity of IPLT for the setting being considered in this work remains open in general. 
    In addition, the capacity of IPLT for the setting in which the coefficient matrix of the required linear combinations is full-rank (but not necessarily MDS) is still open. 
    \item[2)] Characterizing the capacity of IPLT in the presence of a prior side information is another direction for future research. 
    This research direction is motivated by the recent developments in IPIR and IPLC with side information~\cite{HKRS2019,HS2020}.
    Inspired by these works, different types of individual privacy guarantees can be considered for IPLT. 
    For instance, one may need to protect only the identity of every individual message required for the computation (and not the identities of the side information messages); 
    or it may be needed to protect the identity of every individual message which is required for the computation, or belongs to the side information. 
    \item[3)] Another important direction for research is to establish the fundamental limits of the multi-server setting of the PLT problem with individual privacy guarantees. 
    This problem subsumes the problems of multi-server PIR and multi-server PLC with individual privacy guarantees.
    These problems have not been studied yet, and the advantage of the individual privacy requirement over the joint privacy requirement in the multi-server setting of PIR or PLC remains unknown.
\end{itemize}


\appendix[Illustrative Examples of the GPC-PIA Protocol]
In this appendix, we provide three illustrative examples of the GPC-PIA protocol. 
Example~\ref{ex:1} corresponds to a scenario in which $D$ divides $K$, and Examples~\ref{ex:2} and~\ref{ex:3} correspond to scenarios with $L\leq S$ and $L>S$, respectively. 


\begin{example}\label{ex:1}
Consider a scenario in which the server has ${K=24}$ messages, ${X}_1,\dots,{X}_{24}\in\mathbbmss{F}_{17}^{N}$ for an arbitrary integer ${N\geq 1}$, and the user wishes to compute ${L=2}$ linear combinations of ${D=8}$ messages ${X}_2$, ${X}_4$, ${X}_5$, ${X}_7$, ${X}_8$, ${X}_{10}$,
${X}_{11}$, ${X}_{18}$, say, 
\begin{align*}
Z_{1} & =2{X}_2+15{X}_4+3{X}_5+6{X}_7+{X}_8\\ & \quad +4{X}_{10}+11{X}_{11}+13{X}_{18},	\\[0.125cm]
Z_{2} & =6{X}_2+9{X}_4+4{X}_5+3{X}_7+11{X}_8\\ & \quad +15{X}_{10}+13{X}_{11}+8{X}_{18}.
\end{align*}
For this example, 
$\mathrm{W}=\{2,4,5,7,8,10,11,18\}$, and 
\begin{equation*}
\mathrm{V} = 
\begin{bmatrix}
2 & 15 & 3 & 6 & 1 & 4 & 11 & 13\\
6 & 9 & 4 & 3 & 11 & 15 & 13 & 8
\end{bmatrix}.
\end{equation*} 

In this example, $D\mid K$. For such cases, the GPC-PIA protocol reduces to a simple partition-and-code scheme. 
In particular, the blocks $\mathrm{G}_{1},\dots,\mathrm{G}_{n+1}$ are all of the same size $L\times D$, and 
hence the matrix $\mathrm{G}$ will consist of $n+1$ blocks of equal size $L\times D$. 
Note that when $D\mid K$, $\mathrm{G}$ does not have any column-blocks that create partial interference alignment between the row-blocks of $\mathrm{G}$. 



We modify ${\mathrm{W}}$ by randomly permuting the elements in the original set $\mathrm{W}$, and let ${\mathrm{V}}$ be a matrix that is constructed by applying the same permutation on the columns of the original matrix $\mathrm{V}$. 
For this example, suppose that the modified set ${\mathrm{W}}$ and the modified matrix ${\mathrm{V}}$ are respectively given by ${{\mathrm{W}}=\{5,8,11,2,4,7,10,18\}}$, and
\begin{equation*}
{\mathrm{V}} = 
\begin{bmatrix}
3 & 1 & 11 & 2 & 15 & 6 & 4 & 13\\
4 & 11 & 13 & 6 & 9 & 3 & 15 & 8
\end{bmatrix}.
\end{equation*} 


Here, $R= K \pmod D =0$, $S={\gcd(D+R,R)=8}$, $n=\lfloor {K}/{D}\rfloor-1=2$, $m={R}/{S}+1=1$, and ${t={D}/{S}-1=0}$. 
Note that $L=2<S=8$.

For this example, the user's query consists of a $6\times24$ matrix $\mathrm{G}$ and a permutation $\pi$ on $\{1,\dots,24\}$. 
The matrix $\mathrm{G}$ contains three blocks $\mathrm{G}_1,\mathrm{G}_2,\mathrm{G}_3$, each of size $2\times8$, 
\begin{equation*} \label{eq:4.}
\mathrm{G} = 
\begin{bmatrix}
\mathrm{G}_1 & 0_{2\times 8} & 0_{2\times 8}\\
0_{2\times 8} & \mathrm{G}_2 & 0_{2\times 8}\\
0_{2\times 8} & 0_{2\times 8} & \mathrm{G}_3 
\end{bmatrix}.
\end{equation*}


To construct $\mathrm{G}$, the user follows a randomized procedure. That is, the user randomly selects one of the three blocks $\mathrm{G}_1,\mathrm{G}_2,\mathrm{G}_3$ (each with probability ${D}/{K}={1}/{3}$), and takes the selected block to be equal to $\mathrm{V}$. 
For this example, suppose that the user selects the block $\mathrm{G}_{2}$, and then sets $\mathrm{G}_{2}$ equal to ${\mathrm{V}}$. 
To construct the remaining blocks, namely, $\mathrm{G}_{1}$ and $\mathrm{G}_{3}$, the user randomly generates two MDS matrices, each of size $2\times8$. 
For this example, suppose $\mathrm{G}_1$ and $\mathrm{G}_3$ are given by
\begin{equation*}
\mathrm{G}_{1} = 
\begin{bmatrix}
1 & 4 & 7 & 6 & 3 & 12 & 4 & 9\\
5 & 7 & 6 & 9 & 3 & 15 & 2 & 1
\end{bmatrix},
\end{equation*}
\begin{equation*}
\mathrm{G}_{3} =
\begin{bmatrix}
9 & 13 & 2 & 10 & 7 & 1 & 15 & 3\\
9 & 11 & 12 & 3 & 13 & 13 & 7 & 10
\end{bmatrix}.   
\end{equation*} 

Next, the user constructs a permutation $\pi$ on $\{1,\dots,24\}$. 
Note that the columns $9$, $10$, $11$, $12$, $13$, $14$, $15$, $16$ of the matrix $\mathrm{G}$ are constructed based on the columns $1,\dots,8$ of the matrix ${\mathrm{V}}$, respectively, and the columns $1,\dots,8$ of ${\mathrm{V}}$ correspond respectively to the message indices in $\mathrm{W}$, i.e., ${5}$, $8$, ${11}$, $2$, $4$, $7$, ${10}$, ${18}$. 
Thus, the user constructs the permutation $\pi$ such that $\pi(5)=9$, $\pi(8)=10$, $\pi(11)=11$, $\pi(2)=12$, $\pi(4)=13$, $\pi(7)=14$, $\pi(10)=15$, $\pi(18)=16$. 
For $i\not\in \mathrm{W}$, the user then randomly chooses $\pi(i)$ subject to the constraint that $\pi$ forms a valid permutation on $\{1,\dots,24\}$. 

The user sends the matrix $\mathrm{G}$ and the permutation $\pi$ to the server as the query.
Upon receiving the user's query, the server first permutes the rows of the matrix $\mathrm{X} = [X_1^{\transpose},\dots,X_{24}^{\transpose}]^{\transpose}$ according to the permutation $\pi$ to obtain the vector $\tilde{\mathrm{X}}=\pi(\mathrm{X})$, i.e., $\tilde{X}_{\pi(i)} = X_i$ for $i\in \{1,\dots,24\}$. For this example, suppose that the matrix $\tilde{\mathrm{X}}$ is given by
\begin{align*}
& [X_{1}^{\transpose},X_{22}^{\transpose},X_{13}^{\transpose},X_{19}^{\transpose},X_{24}^{\transpose},X_{17}^{\transpose},X_{20}^{\transpose},X_{12}^{\transpose},\\
& \quad \quad
X_{5}^{\transpose},X_{8}^{\transpose},X_{11}^{\transpose},X_{2}^{\transpose},X_{4}^{\transpose},X_{7}^{\transpose},X_{10}^{\transpose},X_{18}^{\transpose},\\
& \quad \quad\quad 
X_{3}^{\transpose},X_{15}^{\transpose},X_{9}^{\transpose},X_{21}^{\transpose},X_{16}^{\transpose},X_{14}^{\transpose},X_{6}^{\transpose},X_{23}^{\transpose}]^{\transpose}.
\end{align*}
The server then computes $\mathrm{Y}=\mathrm{G}\mathrm{\tilde{\mathrm{X}}}$, and sends the matrix $\mathrm{Y}$ back to the user as the answer. 
Let $\Tilde{\mathrm{X}}_{\mathrm{T}_{1}}, \Tilde{\mathrm{X}}_{\mathrm{T}_{2}}, \Tilde{\mathrm{X}}_{\mathrm{T}_{3}}$ denote the first, second, and third ${D=8}$ rows of the matrix $\mathrm{X}$, respectively.
Note that $\tilde{\mathrm{X}}_{\mathrm{T}_1}=[X_{5}^\transpose,X_{8}^\transpose,X_{11}^\transpose,X_{2}^\transpose,X_{4}^\transpose,X_{7}^\transpose,X_{10}^\transpose, X_{18}^\transpose]^{\transpose}=\mathrm{X}_{\mathrm{W}}$ corresponds to the messages required for the user's computation. 
Thus, $\mathrm{G}_{2}\tilde{\mathrm{X}}_{\mathrm{T}_{2}}=\mathrm{V}\mathrm{X}_{\mathrm{W}}$, which is the user's demand matrix. 
Note that $\mathrm{Y} = [\mathrm{Y}_1^{\transpose},\mathrm{Y}_2^{\transpose},\mathrm{Y}_3^{\transpose}]^{\transpose}$, where $\mathrm{Y}_1 \triangleq \mathrm{G}_1 \tilde{\mathrm{X}}_{\mathrm{T}_1}$, $\mathrm{Y}_2 \triangleq \mathrm{G}_2 \tilde{\mathrm{X}}_{\mathrm{T}_2}$, and $\mathrm{Y}_3 \triangleq \mathrm{G}_3 \tilde{\mathrm{X}}_{\mathrm{T}_3}$. 
This implies that the user can recover their demand matrix $\mathrm{V}\mathrm{X}_{\mathrm{W}}$ from $\mathrm{Y}_2$. 

For this example, the GPC-PIA protocol achieves the rate $(\lfloor {K}/{D}\rfloor+ {R}/{S})^{-1} = 1/3$, whereas the optimal JPLT protocol of~\cite{HES2021JointJournal} achieves a lower rate $L/(K-D+L)=2/18$.

\end{example}


\begin{example} \label{ex:2}
\normalfont
Consider  a  scenario in which the server has ${K=24}$ messages, ${X}_1,\dots,{X}_{24}\in\mathbbmss{F}_{17}^{N}$ for any arbitrary ${N\geq1}$, and the user wishes to compute ${L=2}$ linear combinations of ${D=9}$ messages ${X}_2$, ${X}_4$, ${X}_5$, ${X}_7$, ${X}_8$, ${X}_{10}$,${X}_{11}$, ${X}_{18}$, ${X}_{23}$, say, 
\begin{align*}
Z_{1} & =2{X}_2+15{X}_4+3{X}_5+6{X}_7+{X}_8\\ & \quad +4{X}_{10}+11{X}_{11}+13{X}_{18}+9{X}_{23},	\\[0.125cm]
Z_{2} & =6{X}_2+9{X}_4+4{X}_5+3{X}_7+11{X}_8\\ & \quad +15{X}_{10}+13{X}_{11}+8{X}_{18}+{X}_{23}.
\end{align*}

Similarly as in the previous example, we modify the set $\mathrm{W}$ and the matrix $\mathrm{V}$. 
For this example, suppose that the modified set $\mathrm{W}$ and the modified matrix $\mathrm{V}$ are respectively given by $\mathrm{W}=\{10,4,8,11,7,23,18,2,5\}$, and 
\begin{equation*}
\mathrm{V}= 
\begin{bmatrix}
4 & 15 & 1 & 11 & 6 & 9 & 13 & 2 & 3\\
15 & 9 & 11 & 13 & 3 & 1 & 8 & 6 & 4
\end{bmatrix}.
\end{equation*} 

Here, ${R= K \pmod D =6}$, ${S=\gcd(D+R,R)=3}$, $n=\lfloor {K}/{D}\rfloor-1=1$, $m={R}/{S}+1=3$, and ${t={D}/{S}-1=2}$. Note that $L=2<S=3$. 
In this case, ${D\nmid K}$, and a simple partition-and-code based scheme as in Example~\ref{ex:1} cannot be used.  

For this example, the user's query consists of an $8\times24$ matrix $\mathrm{G}$ and a permutation $\pi$ on $\{1,\dots,24\}$. 
The matrix $\mathrm{G}$ is constructed using two blocks $\mathrm{G}_1$ and $\mathrm{G}_2$ of size $2\times9$ and $6\times15$, respectively, 
\begin{equation} \label{eq:4.15}
\mathrm{G} = 
\begin{bmatrix}
\mathrm{G}_1 & 0_{2\times 15}\\
0_{6\times 9} & \mathrm{G}_2
\end{bmatrix},
\end{equation} where the construction of $\mathrm{G}_1$ and $\mathrm{G}_2$ is described below.

The user randomly selects one of the blocks $\mathrm{G}_1,\mathrm{G}_2$, where the probability of selecting $\mathrm{G}_1$ is ${D/K={9}/{24}}$, and the probability of selecting $\mathrm{G}_2$ is ${(D+R)/K={15}/{24}}$. 
Depending on whether $\mathrm{G}_1$ or $\mathrm{G}_2$ is selected, the construction of each of these blocks is different. In this example, suppose the user selects $\mathrm{G}_2$. 
In this case, the user takes $\mathrm{G}_1$ to be a randomly generated MDS matrix of size $2\times 9$, say, 
\begin{equation} \label{eq:4.16}
\mathrm{G}_{1} = \begin{bmatrix}
    3 & 14 & 11 & 8 & 4 & 10 & 5 & 5 & 6\\
    12 & 16 & 3 & 4 & 6 & 3 & 7 & 15 & 4
\end{bmatrix}.
\end{equation}
To construct $\mathrm{G}_{2}$, the user first constructs a $2\times 15$ matrix $\mathrm{C} = [\mathrm{C}_1,\mathrm{C}_2,\mathrm{C}_{3},\mathrm{C}_{4},\mathrm{C}_{5}]$, where the column-blocks $\mathrm{C}_1,\dots,\mathrm{C}_5$, each of size $2\times 3$, are 
constructed as follows. The user partitions the columns of $\mathrm{V}$ into three column-blocks $\mathrm{V}_1,\mathrm{V}_2,\mathrm{V}_3$, each of size $2\times 3$, i.e., 
\[
\mathrm{V}_{1} = 
\begin{bmatrix}
    4 & \hspace{-0.125cm} 15 & \hspace{-0.125cm} 1\\
    15 & \hspace{-0.125cm} 9 & \hspace{-0.125cm} 11
\end{bmatrix},\hspace{0.125cm}
\mathrm{V}_{2} = 
\begin{bmatrix}
    11 & \hspace{-0.125cm} 6 & \hspace{-0.125cm} 9\\
    13 & \hspace{-0.125cm} 3 & \hspace{-0.125cm} 1
\end{bmatrix},\hspace{0.125cm}
\mathrm{V}_{3} = 
\begin{bmatrix}
    13 & \hspace{-0.125cm} 2 & \hspace{-0.125cm} 3 \\
    8 & \hspace{-0.125cm} 6 & \hspace{-0.125cm} 4 
\end{bmatrix}.
\]
The user then randomly chooses three indices $i_1,i_2,i_3$ from $\{1,2,3,4,5\}$, say, $i_1=1$, $i_2=3$, $i_3=5$, and takes $\mathrm{C}_{i_1}=\mathrm{C}_1 = \mathrm{V}_1$, $\mathrm{C}_{i_2} = \mathrm{C}_{3}=\mathrm{V}_2$, $ \mathrm{C}_{i_3} = \mathrm{C}_{5}=\mathrm{V}_3$.
Next, the user takes the remaining column-blocks of $\mathrm{C}$, i.e., $\mathrm{C}_2$ and $\mathrm{C}_4$, to be randomly generated matrices of size $2\times 3$ such that 
$\mathrm{C} = [\mathrm{C}_1,\mathrm{C}_2,\mathrm{C}_{3},\mathrm{C}_{4},\mathrm{C}_{5}]$ is an MDS matrix. For this example, suppose the user takes $\mathrm{C}_2$ and $\mathrm{C}_4$ as 
\begin{equation*}
\mathrm{C}_{2}=
\begin{bmatrix}
1 & 4 & 7\\
5 & 7 & 6\\
\end{bmatrix}, \quad \quad
\mathrm{C}_{4} = 
\begin{bmatrix}
    6 & 3 & 12 \\
    9 & 3 & 15 
\end{bmatrix}.
\end{equation*} Thus, the matrix $\mathrm{C}$ is given by $\mathrm{C} = [\mathrm{V}_1,\mathrm{C}_2,\mathrm{V_2},\mathrm{C_4},\mathrm{V_3}]$. 
The user then randomly chooses $t+m=5$ distinct elements $x_1,x_2,x_3,y_1,y_2$ from $\mathbbmss{F}_{17}$, say, $x_1=1$, $x_2=5$, $x_3=7$, $y_1=11$, $y_2=16$, and constructs a $3\times 2$ Cauchy matrix whose entry $(i,j)$ is given by $\omega_{i,j}\triangleq (x_i-y_j)^{-1}$, i.e., 
\[
\begin{bmatrix}
\omega_{1,1} & \omega_{1,2} \\
\omega_{2,1} & \omega_{2,2} \\
\omega_{3,1} & \omega_{3,2}
\end{bmatrix}
= \begin{bmatrix}
5 & 9\\
14 & 3\\
4 & 15
\end{bmatrix}.
\]
Next, the user constructs the matrix $\mathrm{G}_2$ as 
\begin{align*}
\mathrm{G}_{2} & = 
\begin{bmatrix}
\alpha_1\omega_{1,1}\mathrm{C}_1 & \alpha_{2}\omega_{1,2}\mathrm{C}_{2} & \alpha_{3}\mathrm{C}_{3} & 0_{2\times 3} & 0_{2\times 3}\\
\alpha_1\omega_{2,1}\mathrm{C}_1 & \alpha_{2}\omega_{2,2}\mathrm{C}_{2} & 0_{2\times 3} & \alpha_{4}\mathrm{C}_{4} & 0_{2\times 3}\\
\alpha_1\omega_{3,1}\mathrm{C}_1 & \alpha_{2}\omega_{3,2}\mathrm{C}_{2} & 0_{2\times 3} & 0_{2\times 3} & \alpha_{5}\mathrm{C}_{5}
\end{bmatrix}\\
& = \begin{bmatrix}
5\alpha_1\mathrm{C}_1 & 9\alpha_{2}\mathrm{C}_{2} & \alpha_{3}\mathrm{C}_{3} & 0_{2\times 3} & 0_{2\times 3} \\
14\alpha_1\mathrm{C}_1 & 3\alpha_{2}\mathrm{C}_{2} & 0_{2\times 3} & \alpha_{4}\mathrm{C}_{4} & 0_{2\times 3} \\
4\alpha_1\mathrm{C}_1 & 15\alpha_{2}\mathrm{C}_{2} & 0_{2\times 3} & 0_{2\times 3} & \alpha_{5}\mathrm{C}_{5}
\end{bmatrix},
\end{align*}where the (scalar) parameters $\alpha_1,\dots,\alpha_5$ are chosen  
such that by performing row-block operations on $\mathrm{G}_{2}$, the user can obtain the matrix $[\mathrm{C}_1,0_{2\times 3},\mathrm{C}_3,0_{2\times 3},\mathrm{C}_5]$. 
Note that the second and fourth column-blocks of $\mathrm{G}_{2}$, i.e., the column-blocks that contain scalar multiples of $\mathrm{C}_2$ and $\mathrm{C}_4$, do not contain any column-block of $\mathrm{V}$, and hence must be eliminated by row-block operations. 
Thus, the user randomly chooses the parameters $\alpha_{2}$ and $\alpha_{4}$ (corresponding to the second and fourth column-blocks of $\mathrm{G}_2$) from $\mathbbmss{F}_{17}\setminus{\{0\}}$, say $\alpha_{2}=2$ and $\alpha_{4}=10$. 
The parameters $\alpha_{1}$, $\alpha_{3}$, and $\alpha_{5}$ are chosen as follows. 
To perform row-block operations on $\mathrm{G}_2$, suppose that the user multiplies the first and third row-blocks of $\mathrm{G}_{2}$ by scalars $c_3$ and $c_5$, respectively, and constructs the matrix 
\begin{align*}
& c_3 \begin{bmatrix}5\alpha_1\mathrm{C}_1 & 9\alpha_2\mathrm{C}_2 & \alpha_3\mathrm{C}_3 & 0_{2\times 3} & 0_{2\times 3}\end{bmatrix} \\ & \quad +c_5\begin{bmatrix}4\alpha_1\mathrm{C}_1 & 15\alpha_2\mathrm{C}_2 & 0_{2\times 3} & 0_{2\times 3} & \alpha_5\mathrm{C}_5\end{bmatrix}\vspace{-0.75cm}
\end{align*}
\[\hspace{-0.125cm}
\begin{split}
&= 
[\begin{matrix} (5c_3+4c_5)\alpha_1\mathrm{C}_1 & (9c_3+15c_5)\alpha_2\mathrm{C}_2 \end{matrix} \\
 &\qquad\qquad\qquad\qquad\qquad \begin{matrix} c_3\alpha_3\mathrm{C}_3 & 0_{2\times 3} & c_5\alpha_5\mathrm{C}_5 \end{matrix}]
\end{split}
\]
Thus, the user can recover the matrix $[\mathrm{C}_1,0_{2\times 3},\mathrm{C}_3,0_{2\times 3},\mathrm{C}_5]$ by performing row-block operations on the matrix $\mathrm{G}_2$ so long as 
${(5c_3+4c_5)\alpha_1=1}$, ${9c_3+15c_5=0}$, ${c_3\alpha_3=1}$, and ${c_5\alpha_5=1}$. 
Note that the choice of $\omega_{i,j}$'s to be entries of a Cauchy matrix guarantees that this system of equations has a nonzero solution for all $c_3,c_5,\alpha_1,\alpha_3,\alpha_5$, and the solution is unique for any arbitrary (but fixed) value of $c_3\neq 0$.    
Choosing $c_3$ to be an arbitrary element in $\mathbbmss{F}_{17}\setminus \{0\}$, say, $c_3=1$, the user takes $c_5 = -{9}c_3/{15} = 13$. 
Given $c_3=1$ and $c_5=13$, the user then finds $\alpha_1={1}/({5c_3+4c_5})=3$, $\alpha_3={1}/{c_3}=1$, and $\alpha_5={1}/{c_5}=4$.  
Then, the user constructs $\mathrm{G}_2$ as 
\begin{align}
\mathrm{G}_{2} 
= \begin{bmatrix}
15\mathrm{C}_1 & \mathrm{C}_{2} & \mathrm{C}_{3} & 0_{2\times 3} & 0_{2\times 3}\\
8\mathrm{C}_1 & 6\mathrm{C}_{2} & 0_{2\times 3} & 10\mathrm{C}_{4} & 0_{2\times 3}\\
12\mathrm{C}_1 & 13\mathrm{C}_{2} & 0_{2\times 3} & 0_{2\times 3} & 4\mathrm{C}_{5}\\
\end{bmatrix}. 
\label{eq:4.17}
\end{align}
Combining $\mathrm{G}_{1}$ and $\mathrm{G}_{2}$ given by~\eqref{eq:4.16} and~\eqref{eq:4.17}, the user then constructs the matrix $\mathrm{G}$ as in~\eqref{eq:4.15}.

Next, the user constructs a permutation $\pi$ on $\{1,\dots,24\}$. 
Note that the columns $10$, $11$, $12$, $16$, $17$, $18$, $22$, $23$, $24$ of the matrix $\mathrm{G}$ are constructed based on the columns $1,\dots,9$ of the matrix $\mathrm{V}$, respectively, and the columns $1,\dots,9$ of $\mathrm{V}$ correspond respectively to the message indices ${10}$, $4$, $8$, ${11}$, $7$, ${23}$, ${18}$, $2$, $5$. 
Thus, the user constructs the permutation $\pi$ such that $\pi(10)=10$, $\pi(4)=11$, $\pi(8)=12$, $\pi(11)=16$, $\pi(7)=17$, $\pi(23)=18$, $\pi(18)=22$, $\pi(2)=23$, $\pi(5)=24$. 
For ${i\not\in \{2,4,5,7,8,10,11,18,23\}}$, the user then randomly chooses $\pi(i)$ subject to the constraint that $\pi$ forms a valid permutation on $\{1,\dots,24\}$. 

Then, the user sends the matrix $\mathrm{G}$ and the permutation $\pi$ to the server as the query. 
Upon receiving the user's query, the server first permutes the rows of the matrix $\mathrm{X} = [X_1^{\transpose},\dots,X_{24}^{\transpose}]^{\transpose}$ according to the permutation $\pi$ to obtain the vector $\tilde{\mathrm{X}}=\pi(\mathrm{X})$, i.e., $\tilde{X}_{\pi(i)} = X_i$ for $i\in \{1,\dots,24\}$. For this example, suppose that the matrix $\tilde{\mathrm{X}}$ is given by
\begin{align*}
& [X_{17}^{\transpose},X_{22}^{\transpose},X_{20}^{\transpose},X_{14}^{\transpose},X_{24}^{\transpose},X_{21}^{\transpose},X_{19}^{\transpose},X_{15}^{\transpose},X_{6}^{\transpose},X_{10}^{\transpose},X_{4}^{\transpose},X_{8}^{\transpose},\\
& \quad \quad X_{1}^{\transpose},X_{13}^{\transpose},X_{16}^{\transpose},X_{11}^{\transpose},X_{7}^{\transpose},X_{23}^{\transpose},X_{9}^{\transpose},X_{3}^{\transpose},X_{12}^{\transpose},X_{18}^{\transpose},X_{2}^{\transpose},X_{5}^{\transpose}]^{\transpose}.
\end{align*}
Then the server computes $\mathrm{Y}=\mathrm{G}\mathrm{\tilde{\mathrm{X}}}$, and sends the matrix $\mathrm{Y}$ back to the user as the answer. 
Let $\mathrm{T}_1 = \{1,\dots,9\}$, $\mathrm{T}_2 = \{10,11,12\}$, $\mathrm{T}_3 = \{13,14,15\}$, $\mathrm{T}_4 = \{16,17,18\}$, 
$\mathrm{T}_5 = \{19,20,21\}$, and
$\mathrm{T}_6 = \{22,23,24\}$. 
Note that $[\Tilde{\mathrm{X}}_{\mathrm{T}_2}^{\transpose},\tilde{\mathrm{X}}_{\mathrm{T}_4}^{\transpose},\tilde{\mathrm{X}}_{\mathrm{T}_6}^{\transpose}]^{\transpose} = \mathrm{X}_{\mathrm{W}}$, and $\mathrm{Y} = [\mathrm{Y}_1^{\transpose},\mathrm{Y}_2^{\transpose}]^{\transpose}$, where $\mathrm{Y}_1 \triangleq \mathrm{G}_1 \tilde{\mathrm{X}}_{\mathrm{T}_1}$, and $\mathrm{Y}_2 \triangleq \mathrm{G}_2 [\tilde{\mathrm{X}}_{\mathrm{T}_2}^{\transpose},\tilde{\mathrm{X}}_{\mathrm{T}_3}^{\transpose},\tilde{\mathrm{X}}_{\mathrm{T}_4}^{\transpose},\tilde{\mathrm{X}}_{\mathrm{T}_5}^{\transpose},\tilde{\mathrm{X}}_{\mathrm{T}_6}^{\transpose}]^{\transpose}$.
Let $\tilde{\mathrm{X}}_{\mathrm{T}}\triangleq [\tilde{\mathrm{X}}_{\mathrm{T}_2}^{\transpose},\tilde{\mathrm{X}}_{\mathrm{T}_3}^{\transpose},\tilde{\mathrm{X}}_{\mathrm{T}_4}^{\transpose},\tilde{\mathrm{X}}_{\mathrm{T}_5}^{\transpose},\tilde{\mathrm{X}}_{\mathrm{T}_6}^{\transpose}]^{\transpose}$, and 
let $\mathrm{I}$ be a $2\times 2$ identity matrix. 
Then, the user recovers $[Z^{\transpose}_1,Z^{\transpose}_2]^{\transpose} = \mathrm{V}\mathrm{X}_{\mathrm{W}}$ by computing 
\begin{align*}
& \begin{bmatrix} c_3\mathrm{I} & 0_{2\times2} & c_5\mathrm{I} \end{bmatrix}\mathrm{Y}_2 \\ 
&= \begin{bmatrix} c_3\mathrm{I} & 0_{2\times2} & c_5\mathrm{I} \end{bmatrix} \mathrm{\mathrm{G}}_{2}
\tilde{\mathrm{X}}_{\mathrm{T}}\\
& =\begin{bmatrix} c_3\mathrm{I} & 0_{2\times 2} & c_5\mathrm{I} \end{bmatrix} 
\begin{bmatrix}
15\mathrm{C}_1 &\hspace{-0.125cm} \mathrm{C}_{2} &\hspace{-0.125cm} \mathrm{C}_{3} &\hspace{-0.125cm} 0 & \hspace{-0.125cm}0\\
8\mathrm{C}_1 &\hspace{-0.125cm}\hspace{-0.12cm} 6\mathrm{C}_{2} &\hspace{-0.125cm} 0 &\hspace{-0.125cm} 10\mathrm{C}_{4} &\hspace{-0.125cm} 0\\
12\mathrm{C}_1 &\hspace{-0.12cm} 13\mathrm{C}_{2} &\hspace{-0.125cm} 0 &\hspace{-0.125cm} 0 &\hspace{-0.125cm} 4\mathrm{C}_{5}\\
\end{bmatrix}
\tilde{\mathrm{X}}_{\mathrm{T}}\\
& = \begin{bmatrix}
(15c_3+12c_5)\mathrm{C}_1 & \hspace{-0.125cm} (c_3+13c_5) \mathrm{C}_2  & \hspace{-0.125cm} c_3\mathrm{C}_3 & \hspace{-0.125cm} 0 & \hspace{-0.125cm} 4c_5\mathrm{C}_5 
\end{bmatrix} 
\tilde{\mathrm{X}}_{\mathrm{T}}\\
& = \begin{bmatrix}
\mathrm{C}_1 & 0_{2\times3} & \mathrm{C}_3 & 0_{2\times3} & \mathrm{C}_5 
\end{bmatrix} 
\tilde{\mathrm{X}}_{\mathrm{T}}\\
& = \begin{bmatrix}
\mathrm{V}_{1} & 0 & \mathrm{V}_{2} & 0 & \mathrm{V}_{3} \\
\end{bmatrix} 
\tilde{\mathrm{X}}_{\mathrm{T}}\\
& = \begin{bmatrix}
\mathrm{V}_{1} &\mathrm{V}_{2} &
\mathrm{V}_{3}
\end{bmatrix} 
\begin{bmatrix}
\tilde{\mathrm{X}}_{\mathrm{T}_2}\\
\tilde{\mathrm{X}}_{\mathrm{T}_4}\\
\tilde{\mathrm{X}}_{\mathrm{T}_6} \end{bmatrix} = \mathrm{V}\mathrm{X}_{\mathrm{W}}.
\end{align*} Recall that $c_3=1$ and $c_5=13$. Thus,  $15c_3+12c_5=1$, $c_3+13c_5 = 0$, and $4c_5=1$. 

For this example, the GPC-PIA protocol achieves the rate $(\lfloor {K}/{D}\rfloor+{R}/{S})^{-1} = 1/4$, whereas the optimal JPLT protocol of~\cite{HES2021JointJournal} 
achieves a lower rate $L/(K-D+L)=2/17$. 

\end{example}

\begin{example} \label{ex:3}
\normalfont 
Consider a scenario in which the server has ${K=24}$ messages, ${X}_1,\dots,{X}_{24}\in\mathbbmss{F}_{17}^N$ for any arbitrary ${N\geq1}$, and the user wishes to compute ${L=2}$ linear combinations of ${D=7}$ messages ${X}_2$, ${X}_4$, ${X}_7$, ${X}_{10}$, ${X}_{15}$, ${X}_{18}$, ${X}_{23}$, say,
\begin{align*}
Z_{1} & =2{X}_2+15{X}_4+6{X}_7+4{X}_{10}\\
& \quad +11{X}_{15}+13{X}_{18}+9{X}_{23},\\
Z_{2} & =6{X}_2+9{X}_4+3{X}_7+15{X}_{10}\\
& \quad +13{X}_{15}+8{X}_{18}+{X}_{23}.\\
\end{align*}
For this example, $\mathrm{W}=\{2,4,7,10,15,18,23\}$, and
\begin{equation*}
\mathrm{V} = 
\begin{bmatrix}
2 & 15 & 6 & 4 & 11 & 13 & 9\\
6 & 9 & 3 & 15 & 13 & 8 & 1
\end{bmatrix}.
\end{equation*} 
Similar to the previous examples, we modify the set $\mathrm{W}$ and the matrix $\mathrm{V}$. 
For this example, suppose that the modified set $\mathrm{W}$ and the modified matrix $\mathrm{V}$ are given by $\mathrm{W}=\{10,4,7,23,18,2,15\}$, and 
\begin{equation*}
\mathrm{V} = 
\begin{bmatrix}
15 & 4 & 6 & 9 & 13 & 2 & 11\\
9 & 15 & 3 & 1 & 8 & 6 & 13
\end{bmatrix}.
\end{equation*} 

Here, $R= K \pmod D =3$, ${S=\gcd(D+R,R)=1}$, $n=\lfloor {K}/{D}\rfloor-1=2$, and $m={R}/{L}+1=\frac{5}{2}$. 
Note that ${L=2>S=1}$.


For this example, the user's query consists of a $9\times24$ matrix $\mathrm{G}$ and a permutation $\pi$ on $\{1,\dots,24\}$, constructed as follows. 
The matrix $\mathrm{G}$ is constructed using three blocks $\mathrm{G}_1,\mathrm{G}_2,\mathrm{G}_3$ of size $2\times7$, $2\times 7$, and $5\times 10$, respectively, 
\begin{equation}\label{eq:4.18}
\mathrm{G} = 
\begin{bmatrix}
\mathrm{G}_1 & 0_{2\times 7} & 0_{2\times 10} \\
0_{2\times 7} & \mathrm{G}_2 & 0_{2\times 10} \\
0_{5\times 7} & 0_{5\times 7} & \mathrm{G}_3
\end{bmatrix},
\end{equation} where the construction of  $\mathrm{G}_1,\mathrm{G}_2,\mathrm{G}_3$ is described below.

The user randomly selects one of the blocks $\mathrm{G}_1, \mathrm{G_2},\mathrm{G}_3$, where the probability of selecting $\mathrm{G}_1$ is ${D/K={7}/{24}}$, the probability of selecting $\mathrm{G}_2$ is ${D/K={7}/{24}}$, and the probability of selecting $\mathrm{G}_3$ is ${(D+R)/K={10}/{24}}$. 
Depending on whether $\mathrm{G}_1$, $\mathrm{G}_2$, or $\mathrm{G}_3$ is selected, the construction of each of these blocks is different. 
In this example, we consider the case that the user selects $\mathrm{G}_3$. 
In this case, the user takes $\mathrm{G}_1$ and $\mathrm{G}_2$ to be two randomly generated MDS matrices, each of size $2\times 7$, say, 
\begin{equation}\label{eq:4.19}
\mathrm{G}_{1} = \begin{bmatrix}
    11 & 5 & 10 & 1 & 15 & 2 & 7 \\
    16 & 10 & 16 & 6 & 1 & 1 & 13\\
\end{bmatrix},
\end{equation}
\begin{equation}
\mathrm{G}_{2} = \begin{bmatrix}
    5 & 8 & 14 & 7 & 4 & 3 & 16\\
    3 & 5 & 8 & 1 & 6 & 2 & 15\\
\end{bmatrix}.
\end{equation}
The construction of $\mathrm{G}_3$ is as follows. 
Recall that $\mathrm{V}$ generates a $[7,2]$ MDS code. Thus, the user can obtain the parity-check matrix $\myLambda$ of the MDS code generated by $\mathrm{V}$ as 
\begin{equation*}
\myLambda = 
\begin{bmatrix}
8 & 5 & 9 & 6 & 14 & 11 & 13\\
15 & 6 & 13 & 12 & 6 & 16 & 3\\
9 & 14 & 15 & 7 & 5 & 14 & 2\\
2 & 10 & 16 & 14 & 7 & 8 & 7\\
8 & 12 & 8 & 11 & 3 & 7 & 16\\
\end{bmatrix}.
\end{equation*}
Note that $\myLambda$ itself generates a $[7,5]$ MDS code. Then, the user randomly chooses a $D=7$-subset of $\{1,\dots,10\}$, say, $\{h_1,\dots,h_7\} = \{1,3,4,6,7,8,10\}$, and randomly generates a $2\times 10$ MDS matrix $\mathrm{H}$ such that the submatrix of $\mathrm{H}$ restricted to the columns indexed by $\{h_1,\dots,h_7\}= \{1,3,4,6,7,8,10\}$ is the matrix $\myLambda$. For this example, suppose that the user constructs the matrix $\mathrm{H}$ as
\begin{equation*}
\mathrm{H} = 
\begin{bmatrix}
\mathbf{8} & 1 & \mathbf{5} & \mathbf{9} & 2 & \mathbf{6} & \mathbf{14} & \mathbf{11} & 4 & \mathbf{13} \\
\mathbf{15} & 6 & \mathbf{6} & \mathbf{13} & 3 & \mathbf{12} & \mathbf{6} & \mathbf{16} & 3 & \mathbf{3} \\
\mathbf{9} & 2 & \mathbf{14} & \mathbf{15} & 13 & \mathbf{7} & \mathbf{5} & \mathbf{14} & 15 & \mathbf{2} \\
\mathbf{2} & 12 & \mathbf{10} & \mathbf{16} & 11 & \mathbf{14} & \mathbf{7} & \mathbf{8} & 7 & \mathbf{7} \\
\mathbf{8} & 4 & \mathbf{12} & \mathbf{8} & 8 & \mathbf{11} & \mathbf{3} & \mathbf{7} & 1 & \mathbf{16}
\end{bmatrix}.
\end{equation*} Since $\mathrm{H}$ generates a $[10,5]$ MDS code, it can also be thought of as the parity-check matrix of a $[10,5]$ MDS code. 
The user then takes $\mathrm{G}_3$ to be the generator matrix of the $[10,5]$ MDS code defined by the parity-check matrix $\mathrm{H}$, 
\begin{equation}\label{eq:4.20}
\mathrm{G}_{3} = \begin{bmatrix}
\mathbf{3} & 14 & \mathbf{11} & \mathbf{8} & 4 & \mathbf{10} & \mathbf{8} & \mathbf{5} & 5 & \mathbf{6} \\
\mathbf{12} & 16 & \mathbf{3} & \mathbf{4} & 6 & \mathbf{3} & \mathbf{1} & \mathbf{15} & 8 & \mathbf{4} \\
\mathbf{14} & 11 & \mathbf{7} & \mathbf{2} & 9 & \mathbf{6} & \mathbf{15} & \mathbf{11} & 6 & \mathbf{14} \\
\mathbf{5} & 15 & \mathbf{5} & \mathbf{1} & 5 & \mathbf{12} & \mathbf{4} & \mathbf{16} & 13 & \mathbf{15} \\
\mathbf{3} & 5 & \mathbf{6} & \mathbf{9} & 16 & \mathbf{7} & \mathbf{9} & \mathbf{14} & 14 & \mathbf{10}
\end{bmatrix}.
\end{equation} 
Combining $\mathrm{G}_1,\mathrm{G}_2,\mathrm{G}_3$ given by~\eqref{eq:4.19}-\eqref{eq:4.20}, the user constructs the matrix $\mathrm{G} $ as in~\eqref{eq:4.18}.

Next, the user constructs a permutation $\pi$ on $\{1,\dots,24\}$. 
Note that the columns $15$, $17$, $18$, $20$, $21$, $22$, $24$ of $\mathrm{G}$ are constructed based on the columns $1,\dots,7$ of $\myLambda$; the columns $1,\dots,7$ of $\myLambda$ are constructed based on the columns $1,\dots,7$ of $\mathrm{V}$; and the columns $1,\dots,7$ of $\mathrm{V}$ correspond respectively to the message indices ${4}$, ${10}$, $7$, ${23}$, ${18}$, ${2}$, ${15}$. 
The user then constructs the permutation $\pi$ such that 
$\pi(4)=15$, $\pi(10)=17$, $\pi(7)=18$, $\pi(23)=20$, $\pi(18)=21$, $\pi(2)=22$, $\pi(15)=24$. 
For any $i\not\in \{2,4,7,10,15,18,23\}$, the user then randomly chooses $\pi(i)$ subject to the constraint that $\pi$ forms a valid permutation on $\{1,\dots,24\}$. 
Then, the user sends the matrix $\mathrm{G}$ and the permutation $\pi$ to the server as the query. 

Upon receiving the user's query, the server first permutes the rows of the matrix $\mathrm{X} = [X_1^{\transpose},\dots,X_{24}^{\transpose}]^{\transpose}$ according to the permutation $\pi$ to obtain the matrix $\tilde{\mathrm{X}}=\pi(\mathrm{X})$, i.e., $\tilde{X}_{\pi(i)} = X_i$ for $i\in \{1,\dots,24\}$. For this example, suppose that the matrix $\tilde{\mathrm{X}}$ is given by
\begin{align*}
& [X_{8}^{\transpose},X_{14}^{\transpose},X_{17}^{\transpose},X_{22}^{\transpose},X_{19}^{\transpose},X_{16}^{\transpose},X_{13}^{\transpose},X_{3}^{\transpose},X_{20}^{\transpose},X_{24}^{\transpose},X_{21}^{\transpose},X_{1}^{\transpose},\\ &\quad  X_{6}^{\transpose},X_{12}^{\transpose},X_{4}^{\transpose},X_{5}^{\transpose},X_{10}^{\transpose},X_{7}^{\transpose},X_{9}^{\transpose},X_{23}^{\transpose},X_{18}^{\transpose},X_{2}^{\transpose},X_{11}^{\transpose},X_{15}^{\transpose}]^{\transpose}.
\end{align*}
Then the server computes $\mathrm{Y=\mathrm{G}\mathrm{\tilde{\mathrm{X}}}}$, and sends the matrix $\mathrm{Y}$ back to the user as the answer. 
To recover their demand, the user proceeds as follows.
Let $\mathrm{T}_1 = \{1,\dots,7\}$, $\mathrm{T}_2 = \{8,\dots,14\}$, and $\mathrm{T}_3 = \{15,\dots,24\}$.
Note that $\mathrm{Y} = [\mathrm{Y}_1^{\transpose},\mathrm{Y}_2^{\transpose},\mathrm{Y}_3^{\transpose}]^{\transpose}$, where $\mathrm{Y}_1 \triangleq \mathrm{G}_1 \tilde{\mathrm{X}}_{\mathrm{T}_1}$, $\mathrm{Y}_2 \triangleq \mathrm{G}_2 \tilde{\mathrm{X}}_{\mathrm{T}_2}$, and $\mathrm{Y}_3 \triangleq \mathrm{G}_3 \tilde{\mathrm{X}}_{\mathrm{T}_3}$. 
Then, the user recovers the demand matrix $[Z^{\transpose}_1,Z^{\transpose}_2,Z^{\transpose}_3]^{\transpose} = \mathrm{V}\mathrm{X}_{\mathrm{W}} =$ 
by computing
\begin{align*}
& \begin{bmatrix}
6 & 4 & 13 & 1 & 0\\
0 & 6 & 4 & 13 & 1\\
\end{bmatrix}\mathrm{Y}_3 
\\ & = 
\begin{bmatrix}
6 & 4 & 13 & 1 & 0\\
0 & 6 & 4 & 13 & 1\\
\end{bmatrix} \mathrm{G}_3 \tilde{\mathrm{X}}_{\mathrm{T}_3}\\
& =  
\begin{bmatrix}
\mathbf{15} & 0 & \mathbf{4} & \mathbf{6} & 0 & \mathbf{9} & \mathbf{13} & \mathbf{2} & 0 & \mathbf{11} \\
\mathbf{9} & 0 & \mathbf{15} & \mathbf{3} & 0 & \mathbf{1} & \mathbf{8} & \mathbf{6} & 0 & \mathbf{13} \\
\end{bmatrix} 
\tilde{\mathrm{X}}_{\mathrm{T}_3}\\
& = 
\begin{bmatrix}
{15} & {4} & {6} & {9} & {13} & {2} & {11} \\
{9}  & {15} & {3} & {1} & {8} & {6}  & {13}
\end{bmatrix}  
\begin{bmatrix}
X_{4}\\ X_{10}\\ X_{7}\\ X_{23} \\ X_{18} \\ X_{2} \\ X_{15}
\end{bmatrix} 
=\mathrm{V} \mathrm{X}_{\mathrm{W}}.
\end{align*}
For this example, the GPC-PIA protocol achieves the rate $(\lfloor {K}/{D}\rfloor+{R}/{L})^{-1} = 2/9$, whereas the optimal JPLT protocol of~\cite{HES2021JointJournal} achieves a lower rate $L/(K-D+L)=2/19$.

\end{example}

\bibliographystyle{IEEEtran}
\bibliography{PIR_PC_Refs}

\end{document}